%% file: main.tex
 \documentclass[final,12pt]{clear2026} 

\title[Cost-Aware Optimized Front-Door Experimental Design]{Cost-Aware Optimized Front-Door Experimental Design}
\usepackage[utf8]{inputenc}
\usepackage{amsmath}
\usepackage{amsfonts}
\usepackage{amssymb}
\usepackage{times}
\usepackage{mathrsfs}
\usepackage{tikz}
\usepackage{tcolorbox}
\usepackage{graphicx}
\usepackage{array}
\usepackage{bm}
\usepackage{multirow}

\usetikzlibrary{arrows.meta, positioning}

\newtheorem{model}[theorem]{Model}
\newtheorem{problem}[theorem]{Problem}
\newtheorem{assumption}[theorem]{Assumption}

\newcommand{\E}[1]{\mathrm{E}\!\left[ #1 \right]}
\newcommand{\Var}[1]{\mathrm{Var}\!\left( #1 \right)}
\newcommand{\restatement}[2]{\begin{center}
\fbox{%
  \begin{minipage}{0.97\textwidth}
  \textbf{#1:} #2
  \end{minipage}
} 
\end{center} }

\clearauthor{%
 \Name{Leopold Mareis} \Email{leopold.mareis@tum.de}\\
 \addr 
 Technical University of Munich, Germany
 \AND
 \Name{Mathias Drton} \Email{mathias.drton@tum.de}\\
 \addr 
 Technical University of Munich, Germany \\
\addr Munich Center for Machine Learning, Munich, Germany
}

\jmlrvolume{}  
\jmlryear{}
\jmlrproceedings{}{}
\jmlrpages{}
\jmlrworkshop{}
\editors{}

\begin{document}

\maketitle

\begin{abstract}%
  Causal effect estimation often succeeds cost-constrained sequential data collection. This work considers multivariate linear front-door models with arbitrary unobserved confounding on treatment and response. We optimize the experimental design by balancing the statistical efficiency and measurement costs through partial data. The full‑data efficient influence function for the causal effect is derived, together with the geometry of all observed‑data influence functions. This characterization yields a closed-form optimal sampling policy and an estimator to minimize the asymptotic variance of regular asymptotically linear (RAL) estimators within a class of augmented full-data influence functions. The resulting design also covers back-door estimation. In simulations and applications to biological, medical, and industrial datasets, the optimized designs achieve substantial efficiency gains ($ 5.3 \%$ to $ 31.9 \%$) over naive full-sampling strategies.%
\end{abstract}

\begin{keywords}%
  Causal inference, experimental design, front-door estimation, linear SEM, semiparametric estimation%
\end{keywords}

\section{Introduction}

Estimating the linear causal effect between a treatment $X_t$ and a response $X_r$ from observational data under unobserved confounding is a central problem in causal inference, with applications ranging from pharmacology, medicine, economics, to social sciences. This paper addresses a problem that arises naturally in practice: how to allocate a limited measurement budget across staged sampling designs with intermediate costs so as to minimize the variance of a front-door causal effect estimator. We combine semiparametric efficiency theory and causal graphical modelling to derive RAL estimators and optimized experimental designs that achieve substantial efficiency gains over naive full-measurement strategies without sacrificing inferential validity.

A helpful way to view the setting is as a statistician planning a medical study. The goal is to estimate the effect of a medical procedure, but collecting additional biological measurements per patient is costly, risky or time intensive. As even partial data can be informative, the experimenter must decide which patients to measure more extensively and how to allocate limited measurement resources while maintaining statistical efficiency. In our framework, baseline covariate information allows the experimenter to make patient-specific probabilistic decisions about whether to measure mediators and, conditional on that, whether to measure the response. These decisions are based on information already available at the time of sampling. Importantly, the design operates at the population level: while every unit contributes information in expectation, it is not always optimal to collect all measurements for every individual.

We study these issues in a multivariate linear front-door model defined by structural equations 
\[
X = \beta X + \varepsilon,\quad X = (I - \beta)^{-1} \varepsilon,
\]
where the lower-triangular parameter matrix $ \beta $ encodes linear causal relationships between the confounder, treatment, mediator, and response variables $ X:=X_{C,t,M,r}$ as visualized in Figure \ref{figure_causal_graphs}. To reflect the sequential data collection, we introduce coarsening levels $ X_{C,t}$, $ X_{C,t,M} $, $ X_{C,t,M,r} $ with known propensities $ \pi $, resulting in an observed-data distribution $ \mathcal{P}_{\beta, \varepsilon, \pi} $.

In well-behaved statistical models all sensible consistent estimators, including the MLE, are regular and asymptotically linear (RAL). Our work builds on several strands of RAL literature. \cite{robins1994} and \cite{robins1995} established semiparametric efficiency of parameter estimates in multivariate regression models under coarsening, but were limited to single-stage regressions. Over the years, this single-stage problem has been optimized with recent advances in automatic bias corrections using learned Riesz representers \citep{Tan_IPW, hines2025automaticdebiasingneuralnetworks}. Utilizing the graphical structure in acyclic directed mixed graphs (ADMGs), \cite{bhattacharya2020identification} and \cite{Jung_double_ID} developed efficient, doubly robust full-data estimators targeting interventional distributions \citep{Hines2022, chernozhukov2017double, vanderVaart2014}. In contrast, we impose a linear structure, further restricting the semiparametric tangent space and allowing for practical estimation without non-parametric density estimation. Unifying graphical models and missing data theory, \cite{mohan_missingness} solved the distributional identification problem, allowing for reinforced causal effect estimation \citep{seitzer2021causal}. Related work includes the search for back‑door adjustment sets that minimize the asymptotic variance of linear regression estimators \citep{Witte2020, drton_linear_identifiability}, as well as the optimization of efficient linear regression within instrumental‑variable experimental designs \citep{mareis_2025}.

\paragraph*{Contribution} 
The primary contributions are the following. 
\begin{itemize}
\item \textbf{Theorem~\ref{thm_full_data_EIF}}: We derive the efficient influence function $\varphi_\xi^{F,\mathrm{eff}}$ for the causal effect $\xi$ in the full-data model $\mathcal{M}_1$. This function characterizes the minimum achievable asymptotic variance among all RAL estimators when complete observations are available. It decomposes into two orthogonal components that separately and efficiently target the treatment-to-mediator parameter $\beta_{Mt}$ and the mediator-to-response parameter $\beta_{rM}$.
\item \textbf{Theorem~\ref{theorem_optimization}:} Given $\varphi_\xi^{F,\mathrm{eff}}$, we characterize the geometry of all observed-data influence functions in $\mathcal{M}_\pi$ and identify the optimal augmentation within the class $IF^\pi(\varphi^{F,\mathrm{eff}})$; see Lemma~\ref{lemma_tsiatis}. Theorem~\ref{theorem_optimization} yields a closed-form optimized propensity $ \pi^\ast $ that solves the constrained minimization of asymptotic variance over all RAL estimators within $IF^\pi(\varphi^{F,\mathrm{eff}})$ and all sampling regimes under a fixed expected budget $b_0$. The optimal propensity balances the statistical efficiency of the intermediate estimators $\hat\beta_{Mt,n}$ and $\hat\beta_{rM,n}$ against measurement cost. The optimized design for back-door estimation follows as a corollary. 
\end{itemize}
In a simulation study and in applications to biological, medical and industrial datasets, we demonstrate substantial efficiency gains of $ 5.3 \% $ to $ 31.9 \% $ compared to the naive full-sampling design in Section~\ref{section_Computational_results}. To ensure reproducibility and support practical applications through step‑by‑step instructions, the accompanying code is available at \href{https://doi.org/10.5281/zenodo.18960268}{https://doi.org/10.5281/zenodo.18960268}.

\section{Multivariate Front-Door Model} 
We begin by introducing the linear multivariate front-door model. 
Let $ (\Omega, \mathcal{F}, \mathcal{P}) $ be a probability space, and consider the Hilbert space $ \mathcal{H} = L^2_0(\Omega; \mathbb{R}^d) $ consisting of all $ d $~dimensional mean-zero, finite variance functions $ h: \Omega \to \mathbb{R}^d $, with its inner product $ \E{ h_1^\top h_2} $. Let $ \mathcal{C}(A,B) $ denote the space of continuous functions from $ A $ to $ B $. 

\begin{definition}\label{definition_ANMFD}
Fix integers $ d_C, d_M\in \mathbb{N} $ and let $ d := d_C + 1 + d_M + 1 $. We partition the vector $ x \in \mathbb{R}^d $ as $ x = (x_C, x_t, x_M, x_r) $. The parameter set of admissible coefficient matrices is
\[
\mathcal{B} := \left\{ \beta \in \mathbb{R}^{d \times d} : \beta = \begin{pmatrix} 0 & 0 & 0 & 0  \\ \beta_{tC} & 0 & 0 & 0 \\ \beta_{MC} & \beta_{Mt} & 0 & 0  \\ \beta_{rC} & 0 & \beta_{rM} & 0  \end{pmatrix} \right\}.
\]
The nuisance parameter space $ \mathcal{E} $ of additive noise terms is restricted by a factorization of the induced measures $ P_\varepsilon $. In particular, it is defined as
\begin{align} \label{formula_varepsilon_factorization}
\mathcal{E} := \left\{\varepsilon \in \mathcal{H} \, : \ \begin{matrix} 
\varepsilon  \text{ absolutely continuous, } \\ 
\mathcal{P}_\varepsilon = \mathcal{P}_{\varepsilon_C} \otimes \mathcal{P}_{\varepsilon_{t,r}} \otimes \mathcal{P}_{\varepsilon_M} 
\end{matrix} \right\}.
\end{align}
The linear multivariate front-door model is then defined as the set of pushforward measures
\[
\mathcal{M}_1 = \{ \mathcal{P}_{\beta, \varepsilon} =  (I - \beta)^{-1} \mathcal{P}_\varepsilon \, : \, \beta \in \mathcal{B}, \varepsilon \in \mathcal{E} \},
\]
which satisfy the linear structural equation  $ X = \beta X + \varepsilon $, or equivalently $ X = (I - \beta)^{-1} \varepsilon $.
\end{definition}
This formulation allows for confounded errors $ \varepsilon_t $ and $ \varepsilon_r $, as well as arbitrary confounding and dependencies within the subvectors $ X_C $ and $ X_M $. As motivated in the introduction, we further permit the data to be generated by a known coarsening or, more precisely, missingness process, affecting the mediating, and treatment variables $ (X_M, X_t) $ in two stages. The \textit{full-data} model $ \mathcal M_1 $ is particularly important in the subsequent derivation of efficient estimators.

\begin{definition} \label{definition_observed_data}
Let $ X $ follow model $ \mathcal{M}_1 $. For given propensity functions $ \pi_1 \in \mathcal{C}(\mathbb{R}^{d_C + 1}, (0, 1]) $ and $ \pi_2 \in \mathcal{C}(\mathbb{R}^{d_C + 1 + d_M}, (0, 1]) $, let the random variable $ \Delta: \Omega \mapsto \{ 1, 2, \infty \} $, with probabilities
\begin{align*} 
\mathcal{P}_\pi(\Delta = 1 \,|\, X) = 1 - \pi_1(X_{C,t}), &&& 
\mathcal{P}_\pi(\Delta = 2 \,|\, X) = \pi_1(X_{C,t}) (1 - \pi_2(X_{C,t,M})) ,\\ 
&&&\mathcal{P}_\pi(\Delta = \infty \,|\, X) = \pi_1(X_{C,t}) \pi_2(X_{C,t,M}) 
\end{align*}
indicate which subset of $ X $ is observed. The observed-data distribution is then denoted by $ \mathcal{P}_{\theta,\pi}$, with $\theta = (\beta, \varepsilon) $, and generates realizations 
\[
\left\{ \left(\Delta^{(i)}, G_{\Delta^{(i)}}(X^{(i)}) \right) \right\}_{i=1, \dots ,n} , \quad \text{where} \quad 
G_{\Delta^{(i)}}(X^{(i)}) =   \begin{cases} (X_{C}^{(i)}, X_t^{(i)})  & \text{if } \Delta^{(i)} = 1 \\ (X_C^{(i)}, X_t^{(i)}, X_M^{(i)})  & \text{if } \Delta^{(i)} = 2 ,\\ X^{(i)}  &  \text{if } \Delta^{(i)} = \infty. \end{cases}
\]
\end{definition}
\begin{model} 
For each $ \pi_1 \in \mathcal{C}(\mathbb{R}^{d_C + 1}, (0, 1])$ and $ \pi_2 \in \mathcal{C}(\mathbb{R}^{d_C + 1 + d_M}, (0, 1])$, define the observed-data multivariate front-door model as the collection of measures
\[
\mathcal{M}_{\pi} = \left\{ \mathcal{P}_{\theta,\pi} = \mathcal{L}(\Delta, G_\Delta(X)) : \ \begin{matrix}
    \theta = (\beta, \varepsilon) \in \mathcal{B} \times \mathcal{E} = \Theta, \\ 
    X \sim \mathcal{P}_{\beta, \varepsilon} \in \mathcal{M}_1, \\ 
    \Delta | X \sim \mathcal{P}_\pi(\cdot \, | \, X)    
\end{matrix} \right\}.
\]
\end{model}
A visualizations of the causal graph to which the model $ \mathcal{M}_\pi $ is Markov is shown in Figure~\ref{figure_causal_graphs}. 
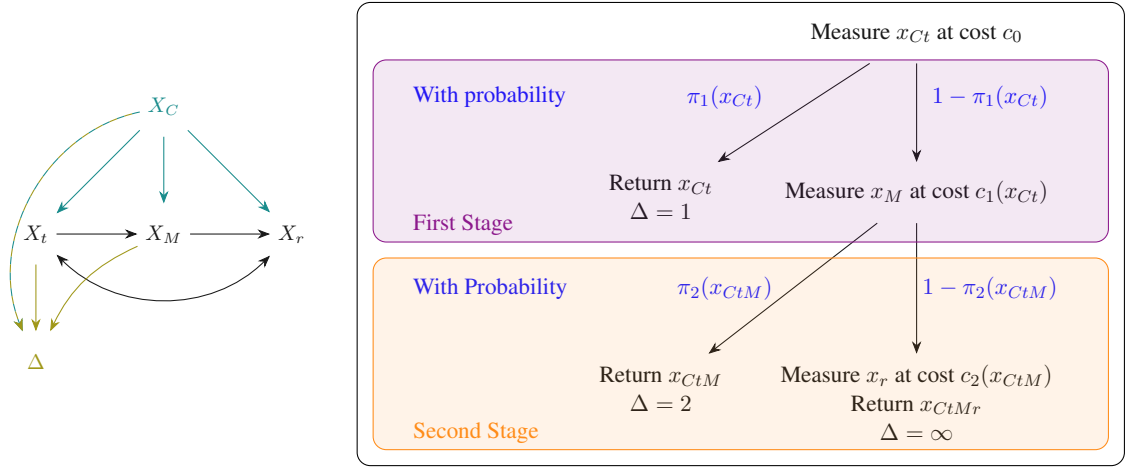
\begin{figure}
    \centering
    \input{rebuttal/Figure_generation}
    \caption{\textit{Left:} ADMG for the multivariate linear front-door model.\textit{ Right:} Generation of a new sample. The {\color{violet} first stage} {\color{blue}probabilistically} decides whether to measure $ x_M $. If so, the {\color{orange} second stage} decides probabilistically whether to additionally measure $ x_r $.}
    \label{figure_causal_graphs}
\end{figure}

\subsection{Practical Workflow}
Figure~\ref{figure_causal_graphs} illustrates the mechanism used to generate a new sample. We write '\textit{Measure }$x_i$' to denote a measurement of $X_i =[ (I - \beta)^{-1}\varepsilon]_i$.  A dataset with four samples from model $ \mathcal{M}_\pi$ can therefore take the form 
\(
\left\{ (\Delta^{(1)} = \infty, x_{CtMr}^{(1)}),\ (2, x_{CtM}^{(2)}),\ (1, x_{Ct}^{(3)}),\ (1, x_{Ct}^{(4)})\right\},
\) and no model can produce a sample on the components $x_{Ctr}$.  The theoretically important full-data Model $ \mathcal{M}_1$ with propensity $ \pi \equiv (1, 1) $ always yields datasets of the form $  \{(\infty, x_{CtMr}^{(i)})\}_{i = 1}^n $. A practical example is personalized staged measurement of patient features for medical diagnostics \citep{vonkleist2025AFA}.

The experimenter adopts model $ \mathcal{M}_\pi $ based on the unknown design matrix $ \beta $ and errors $\varepsilon$ with the goal of estimating the causal effect $\xi$ as cost-efficiently as possible. For each propensity $ \pi = (\pi_1, \pi_2)$, the experimenter can generate datasets according to Figure \ref{figure_causal_graphs} under Model~$ \mathcal{M}_\pi$. The experimenter specifies the base cost $ c_0 $ and cost functions $ c_1, c_2 $, based on real-world considerations, which together determine the cost of each sample. Varying $ \pi $ alters the proportion of $ X_{Ct}, X_{CtM}$ and $X_{CtMr}$ occurrences and thus the associated cost. Since certain parameter configurations require relatively more information on $ X_{Ct} $ and less on $ X_{CtMr}$ for estimating $ \xi $, the experimenter can exploit the staged cost and optimize $ (\pi^\ast_1, \pi^\ast_2) $ via Theorem~\ref{theorem_optimization} or Corollary~\ref{corollary_back_door}. Data is then sampled exclusively from the induced model $ \mathcal{M}_{\pi^\ast}$ to estimate the causal effect in a cost-efficient manner.

\section{Background on Efficient Estimation}
Our objective is to identify propensity functions $ \pi_1, \pi_2 $ that minimize the asymptotic variance of an estimator for the causal effect of the treatment $ X_t $ on the response $ X_r $. This causal effect is defined as the derivative of the interventional expectation evaluated at the treatment level 
\[
\xi(x_t) =\left. \frac{\partial}{\partial x_t^\ast}\mathrm{E}\left[ X_r ; do(X_t = x_t^\ast) \right] \right|_{x_t^\ast = x_t}.
\]
The $ do(\cdot) $ operator modifies the structural equation $ X = \beta X + \varepsilon $ in Definition \ref{definition_ANMFD}, thereby altering the data-generating process \citep{maathuis_handbook_2019}. In the linear front-door model, the causal effect reduces to the sum over all mediating paths, 
\[ 
\xi = \phi(\beta_{Mt}, \beta_{rM}^\top) := \beta_{rM} \beta_{Mt} \in \mathbb{R}. 
\]
We restrict our attention to estimators of $ \xi $ that are regular and asymptotically linear (RAL) in $ \mathcal{M}_\pi $, denoted by $ \mathcal{E}_{\textrm{RAL}}(\xi, \mathcal{M}_\pi) $. Let $ X^{(i)} $, $ i \in [n] $, be i.i.d. draws from $ \mathcal{P}_{\theta, \pi} \in \mathcal{M}_\pi $. An estimator $ \hat{\xi}_n = (X^{(1)}, \dots, X^{(n)})  $ for the causal effect $ \xi $ is asymptotically linear if
\[
 \sqrt{n} (\hat{\xi}_n  - \xi) = n^{-1/2} \sum_{i = 1}^n \varphi(X^{(i)}) + o_{\mathcal{P}_{\theta; \pi}}(1) ,
\] 
for some influence function $ \varphi $ in the Hilbert space of centered, square-integrable  functions. For any RAL estimator $ \hat{\xi}_n $, the quantity $ \sqrt{n} (\hat{\xi}_n  - \xi) $ converges in distribution to $ \mathcal{N}(0, \mathrm{E}_{\mathcal{P}_{\theta_0, \pi}}[ \varphi \varphi^\top ])$, and this convergence is stable under all local perturbations $ \mathcal{P}_{\theta_0 + h / \sqrt{n}} $, $ h \in \Theta$, of the data generating distribution \citep[Thm. 25.20]{Vaart_1998}. 
To ensure regularity, consider parametric submodels $ \{ \mathcal{P}_{(\beta, \varepsilon_0 + \gamma h), \pi} : \gamma \in \mathbb{R} \} \subset \mathcal{M}_\pi $ passing through the truth $\mathcal{P}_{\beta_0, \varepsilon_0} $, for fixed $ h \in \mathcal{E} $.
\begin{assumption}[\cite{newey1990}, A1, or \cite{Vaart_1998}, Thm. 7.10] \label{assumption_regularity}
For each parametric submodel, let $ p_{\beta, \varepsilon_0 + \gamma h }$ denote the density function. Assume that all mappings $ (\beta, \gamma) \mapsto \sqrt {p_{\beta, \varepsilon_0 + \gamma h }(x)}$ are almost surely continuously differentiable at $ \theta_0 $, implying differentiability in quadratic mean. Further assume that the information matrix $ \E{ S_{\beta, \gamma} S_{\beta, \gamma}^\top }$ is finite and positive definite at $ \theta_0 $. 
\end{assumption}
Influence functions must satisfy orthogonality conditions with respect to the semiparametric tangent space $ \mathscr{T} $.  The semiparametric tangent space is, by definition, the $ L^2 $ closure of submodel tangent spaces, where a submodel tangent space is the linear span of its score function $ S_{\beta, \gamma}$ at $ \theta_0 $.
All influence functions are orthogonal to the nuisance tangent space $ \Lambda_\varepsilon $, constructed as the component of $ \mathscr{T} $ associated with the nuisance variation $ \varepsilon $. 
This ensures that influence functions are sensitive only to perturbations of the parameter of interest. 
Consequently, the set of influence functions is characterized by the affine space $ \varphi_0 + \mathscr{T}^\perp$, anchored at an arbitrary influence function $ \varphi_0 $ \citep[Thm. 4.3.]{Tsiatis_Semiparametric}. The efficient influence function minimizes the norm, or equivalently the asymptotic variance (with variances $ V_1 \preceq V_2 $ if $ V_2 - V_1 $ is positive semi-definite), and is characterized by the projection $ \varphi^{\textrm{eff}} = \Pi(\varphi_0 \vert \mathscr{T})$. 
This leads to the following overarching optimization problem.

\begin{problem} \label{problem_main}
Let $ c_0 > 0 $ denote the cost of measuring $ X_{C,t} $, and let the measurable cost functions $ c_1 \in \mathcal{C}(\mathbb{R}^{d_C + 1}, \mathbb{R}_{> 0})$,  $c_2 \in  \mathcal{C}(\mathbb{R}^{d_C + 1 + d_M}, \mathbb{R}_{> 0})$ represent the additional cost of measuring $ X_{M} $ after $ X_{C,t} $ is observed, and of measuring $ X_{r} $ after $ X_{C,t,M} $ is observed. For an average per-sample budget $ b_0 \in \left( c_0, c_0 + \E{c_1(X_{C,t}) + c_2(X_{C,t,M}) }\right] $, the goal is to minimize the asymptotic variance $ \mathrm{Var}_\infty( \hat{\xi}_n ; \pi) $ over all RAL estimators $\hat{\xi}_n$ of $ \xi $ and all sampling regimes induced by $ \mathcal{M}_{\pi = (\pi_1, \pi_2)} $:
\newcommand{\argminexprxi}{%
  \underset{
    \substack{
      \pi_1 \in \mathcal{C}(\mathbb{R}^{d_C + 1}) \\
      \pi_2 \in \mathcal{C}(\mathbb{R}^{d_C + d_M + 1}) \\
      \hat{\xi}_n \in \mathcal{E}_{\textrm{RAL}}(\xi, \mathcal{M}_\pi)
    }
  }{\operatorname{argmin}}
  \; \mathrm{Var}_\infty(\hat{\xi}_n;\pi)
}

\[
\begin{array}{c c l}
  \multirow{3}{*}{$\argminexprxi$}
  & \text{s.t.} & 0 < \pi_1(X_{C,t}) \le 1 \text{ a.e.,} \\
  &             & 0 < \pi_2(X_{C,t,M}) \le 1 \text{ a.e.,} \\
  &             & \mathrm{E}[c_0 + (\pi_1 c_1)(X_{C,t})
                  + \pi_1(X_{C,t})(\pi_2 c_2)(X_{C,t,M})] = b_0.
\end{array}
\]
\end{problem}
The optimal propensity functions $(\pi^\ast_1, \pi^\ast_2) $ specify the experimental design. Since the problem cannot be solved in practice without prior knowledge, we require access to an initial estimate of $ \mathcal{P}_{\theta, \pi} $.

\section{Optimized Causal Effect Estimation}
Since the asymptotic variance $\mathrm{Var}_\infty(\hat{\xi}_n; \pi) = \mathrm{E}[ \varphi \varphi^\top ] = \Var{\varphi}$ fully determines the limiting distribution of any RAL estimator, solving Problem \ref{problem_main} reduces to analyzing the influence function space and deriving the efficient influence function. A logical overview of the proceeding results is presented in Appendix \ref{appendix_structure_of_results}, clarifying the relationships between the full‑data and observed‑data influence functions established by the main theorems and lemmas.

\subsection{Efficient Influence Function in \texorpdfstring{$\mathcal{M}_1$}{M\_1} and its Optimal Augmentation in \texorpdfstring{$\mathcal{M}_\pi$}{M\_pi}} \label{section_eIF_opt_aug}
As a first step towards solving Problem~\ref{problem_main}, we characterize the minimum-variance RAL estimator under the full-data model $ \mathcal{M}_1 $. All major proofs are presented in Appendix \ref{appendix_proofs}.
\newcommand{\TextTheoremFullDataEIF}{Denote the residual error $ \varepsilon_r - \E{ \varepsilon_r | \varepsilon_t} $ by $\varepsilon_r^\perp $. In Model $\mathcal{M}_1$, the $0$-indexed row and column of $ \beta \in \mathcal{B} $ are structurally zero, and the efficient full-data influence function $ \varphi^{F, \text{eff}}$ therefore only has non-zero components on the $(t,M,r)\times(C,t,M)$ block. This remaining block is given by
\[ \resizebox{\textwidth}{!}{ $
  \begin{aligned}
  & \varphi^{F, \text{eff}}_{\beta_{(t,M,r)(C,t,M)}} =  \\ 
& \begin{pmatrix}
\varepsilon_t \varepsilon_C^\top \Var{\varepsilon_C}^{-1} & 0 & 0  \\ 
\varepsilon_M ( \varepsilon_C^\top \Var{\varepsilon_C}^{-1} - \varepsilon_t \Var{\varepsilon_t}^{-1} \beta_{tC}) &  \varepsilon_M \varepsilon_t \Var{\varepsilon_t}^{-1} & 0 \\ 
\varepsilon_r^\perp (\varepsilon_C^\top \Var{\varepsilon_C}^{-1} - \varepsilon_M^\top \Var{\varepsilon_M}^{-1}(\beta_{Mt}\beta_{tC} + \beta_{MC})) & 0 &  \varepsilon_r^\perp \varepsilon_M^\top \Var{\varepsilon_M}^{-1}  \\ 
\end{pmatrix}. 
  \end{aligned}
$ }\]
Consequently, the full-data efficient influence function in $ \mathcal{M}_1 $ for the causal effect $ \xi $ is
\[  \varphi_{\xi}^{F, \text{eff}}  = \beta_{rM} \varepsilon_M \varepsilon_t \Var{\varepsilon_t}^{-1} + \varepsilon_r^\perp \varepsilon_M^\top \Var{\varepsilon_M}^{-1} \beta_{Mt} .\]
}
\begin{theorem} \label{thm_full_data_EIF}
\TextTheoremFullDataEIF
\end{theorem}
The proof decomposes the full-data efficient influence function $\varphi_{\xi}^{F, \text{eff}} $  into two orthogonal components, efficiently and separately targeting $ \beta_{Mt} $ and $ \beta_{rM} $. To transition from full-data to observed-data influence functions, we rely on an existing characterization of their geometry: Each observed-data influence function can be expressed as an augmentation of a full-data influence function.

\begin{lemma}[\cite{Tsiatis_Semiparametric}, Thm. 8.3.] \label{lemma_tsiatis}
Denote the tangent space in Model $ \mathcal{M}_\pi $ corresponding to variation in the propensity functions $ \pi $ by $ \Lambda_\pi $. The observed-data augmentation space is
\begin{align*}
\Lambda_{2,\pi} = &\left\{\sum_{\delta \in \{ 1, 2\}}  \frac{I(\Delta = \delta) - (1 - \pi_\delta)I(\Delta \geq \delta)}{\Pi_{j = 1}^\delta \pi_j} (h_\delta \circ G_\delta)  \, : \, \right. \\ 
&\qquad \qquad \left. h_1 \in L_0^2(\mathbb{R}^{d_C + 1}, \mathbb{R}), h_2 \in L_0^2(\mathbb{R}^{d_C + 1 + d_M}, \mathbb{R}) \right\}. 
\end{align*}
Each full-data influence function $ \varphi^F $ induces an functional family of augmented observed-data influence functions in Model $ \mathcal{M}_\pi $, denoted by
\[
IF^\pi(\varphi^F) =  \left\{ \left[\frac{I(\Delta = \infty)\varphi^F}{\pi_1 \pi_2} + h \right] - \Pi\left([\, \cdot \,] |  \Lambda_\pi \right) \, : \,  h \in \Lambda_{2,\pi} \right\},
\]
The union of these sets coincides with the entire space of observed-data influence functions.
\end{lemma}
This characterization enables us to identify the optimal full-data influence function $ \varphi^{\text{opt},\pi} $ within $ IF^\pi(\varphi^{F, \text{eff}}) $. It can be obtained by retaining the conditional mean term $ \mathbb{E} \left[ \varphi^{F, \text{eff}} | \varepsilon_C, \varepsilon_t \right] $ and weighting the remaining augmentation terms by indicators of the sampling stage variable $ \Delta $ and the inverse propensity functions $ \pi $.

\newcommand{\TextLemmaObservedDataEIF}{Let $ \pi_1 \in \mathcal{C}(\mathbb{R}^{d_C + 1}, (0, 1])$ and $ \pi_2 \in \mathcal{C}(\mathbb{R}^{d_C + d_M + 1}, (0, 1]) $ be given. Under the Model $ \mathcal{M}_\pi $, the optimal observed‑data influence‑function augmentation of the full‑data efficient influence function for the $(t,M,r)\times(C,t,M)$ block of the parameter matrix $\beta $ and for the causal effect $ \xi $~are
\begin{align*}
\varphi_{\beta_{(t,M,r)(C,t,M)}}^{\text{opt},\pi}(X) & = \begin{pmatrix}
\text{diag}(1) & 0 & 0  \\
0 & \frac{\text{diag}\left(I(\Delta \geq 2)\right)}{\pi_1(X_{C,t})}& 0 \\
0 & 0 & \frac{\text{diag}\left(I(\Delta = \infty)\right)}{\pi_1(X_{C,t})\pi_2(X_{C,t,M})} \\
\end{pmatrix} \varphi^{F,\text{eff}}_{\beta_{(t,M,r)(C,t,M)}}(X), \quad \text{and } \\ 
\varphi_{\xi}^{\text{opt},\pi}(X) & = \frac{I(\Delta \geq 2)\beta_{rM} \varepsilon_M \varepsilon_t \Var{\varepsilon_t}^{-1} }{\pi_1(X_{C,t})} + \frac{I(\Delta = \infty)\varepsilon_r^\perp \varepsilon_M^\top \Var{\varepsilon_M}^{-1} \beta_{Mt}}{\pi_1(X_{C,t})\pi_2(X_{C,t,M})}.
\end{align*}
}
\begin{lemma} \label{lemma_observed_data_EIF}
\TextLemmaObservedDataEIF
\end{lemma}
The optimized influence function $ \varphi^{\text{opt},\pi} $ might not be an efficient influence function in the observed-data Model $ \mathcal{M}_\pi $. Recall that every full-data influence function lies in the affine linear space $ \varphi^{F,\text{eff}} + \Lambda_\varepsilon^\perp $ as specified by Equation \eqref{formula_nuisance_orthogonal_compelement}. By linearity of the operator $ \mathcal{J} $, which maps a full-data influence function $ \varphi^F $ to its optimum in $ IF^\pi(\varphi^F) $, the efficient observed-data influence function is given as $ \varphi^{F,\text{eff}} + u $, where the orthogonal nuisance tangent space element $ u \in \Lambda_\varepsilon^\perp $ is chosen to minimize the variance 
\( \Var{\mathcal{J}(\varphi^\text{F,eff})} + \Var{\mathcal{J}(u)} + 2 \mathrm{Cov}\left(\mathcal{J}(u), \mathcal{J}(\varphi^\text{F,eff})\right) \). Nevertheless, the RAL estimator associated with $ \varphi^{\text{opt},\pi} $ achieves an asymptotic variance no larger than that of any other RAL estimator arising from an augmentation in $ IF^\pi(\varphi^{F, \text{eff}}) $. We refer to Appendix \ref{appendix_structure_of_results} for a visualization.

\subsection{Optimized Observed-Data Estimator \texorpdfstring{$\hat{\xi}_n$}{xi\_n} and Optimized Propensity \texorpdfstring{$\pi^{\text{opt}}$}{pi\textasciicircum opt}} \label{section_optimized_EIF}
Efficient RAL estimators $ \hat{\beta}_{Mt}  $ and $ \hat{\beta}_{rM} $ are obtained by solving estimating equations associated with the efficient influence function. Concretely, they correspond to weighted linear equations on the partial datasets defined by the indicators $ I(\Delta^{(i)} \geq 2) $ and $ I(\Delta^{(i)} = \infty) $, respectively. Knowing the optimal estimator for each model $ \mathcal{M}_\pi $ allows solving a restricted version of Problem \ref{problem_main} on the subset of observed-data augmentations of $ \varphi^{F, \text{eff}} $, in particular $ IF^\pi(\varphi^{F, \text{eff}})$. This restriction is purely technical, and Figure \ref{figure_proof_logic} in Appendix \ref{appendix_structure_of_results} illustrates its implied admissible influence function set.
\newcommand{\TextLemmaEfficientEstimator}{The observed-data optimized RAL estimators $ \hat{\beta}_{Mt, n}  $ and $ \hat{\beta}_{rM, n} $ are solutions to a sequence of linear equations on the full and partially observed datasets. Their product estimator $ \hat\xi_n = \hat{\beta}_{rM, n} \hat{\beta}_{Mt, n}  $ is the optimized RAL estimator for the causal effect $ \xi $ in Model $ \mathcal{M}_\pi $. The exact formulas are provided in Appendix~\ref{appendix_proof_nested_least_squares}.}
\begin{lemma} \label{lemma_linear_estimator}
 \TextLemmaEfficientEstimator
\end{lemma}
Note that the equations determining the estimators $ \hat\beta_{rC,n} $ and $ \hat\beta_{rM,n}$ are conceptually different from a sequence of weighted linear regressions. 
\newcommand{\TextLemmaEffectEIFVar}{
The asymptotic variance  $ \mathrm{Var}_\infty(\hat{\xi}_n;\pi) $ of the optimized estimator $ \hat\xi_n $ for the causal effect $ \xi $ in Model $ \mathcal{M}_\pi $,  $ \pi_1 \in \mathcal{C}(\mathbb{R}^{d_C + 1}, (0, 1]) $, $ \pi_2 \in \mathcal{C}(\mathbb{R}^{d_C + d_M + 1}, (0, 1]) $, is determined by
\[ \Var{\varphi_\xi^{\text{opt},\pi}} = \E{ \frac{\varepsilon_t^2 \beta_{rM}  \Var{\varepsilon_M} \beta_{rM}^\top }{ \Var{\varepsilon_t}^{2} \pi_1(X_{C,t})} } +  \E{ \frac{\left(\varepsilon_M^\top \Var{\varepsilon_M}^{-1} \beta_{Mt}\right)^2 \Var{ \varepsilon_r\,  | \, \varepsilon_{t} } }{\pi_1(X_{C,t})\pi_2(X_{C,t,M})}} .
\]}
\begin{lemma} \label{lemma_effect_EIF_var}
\TextLemmaEffectEIFVar
\end{lemma}
\newcommand{\TextTheoremOptimization}{Let $ X $ follow model $ \mathcal{M}_\pi $. Let the base cost $ c_0 $, cost functions $ c_1, c_2 $ as well as the average budget $ b_0 $ be according to the specifications of Problem \ref{problem_main}. Under the additional restriction that admissible influence functions $\varphi $ satisfy $ \varphi \in IF^\pi(\varphi^{F, \text{eff}}) $, Problem \ref{problem_main} is uniquely solvable almost everywhere. Define the leverages \( g_1(X_{C,t}) := \beta_{rM} \Var{\varepsilon_M} \beta_{rM}^\top\frac{\varepsilon_t^2}{\Var{\varepsilon_t}^2} \), and \( g_2(X_{C,t,M}) := \Var{ \varepsilon_r \mid \varepsilon_{t} } \left(\varepsilon_M^\top \Var{\varepsilon_M}^{-1} \beta_{Mt}\right)^2 \). Then the optimal propensity $ \pi^\ast$ is given by \[  (\pi_1^\ast, \pi_2^\ast) =
	 \begin{cases}
	  \left(\min\left(1, \max\left( \sqrt\frac{g_1}{\lambda c_1}, \sqrt{\frac{g_1 + \E{g_2| \varepsilon_{C,t}}}{\lambda(c_1 + \E{c_2| \varepsilon_{C,t}})}}\right)\right), \min\left(1, \sqrt\frac{g_2 c_1}{g_1 c_2}\right) \right), & g_1 < \lambda c_1 , \\ 
	  \left(1, \min\left(1, \sqrt{\frac{g_2}{\lambda c_2}}\right)\right), & g_1 \geq \lambda c_1 ,\\ 
 \end{cases}
	 \] where the constant $ \lambda > 0 $ is chosen to satisfy the problem's budget constraint \[ \mathrm{E}[c_0 + (\pi_1 c_1)(X_{C,t}) + \pi_1(X_{C,t})(\pi_2 c_2)(X_{C,t,M})] = b_0. \]}
\begin{theorem} \label{theorem_optimization} 
\TextTheoremOptimization
\end{theorem}
Several practical considerations accompany this optimality result. First, the conditional expectation $\E{g_2|\varepsilon_{C,t}} $ in the first-stage propensity $ \pi_1^\ast $ can be approximated by sampling from the empirical distribution function $ \hat{F}_n(\varepsilon_M) $. Second, when the solution $ \pi^\ast $ lies in the interior of the feasible set, the budget constraint admits a closed form expression for the Lagrange multiplier,
\[\lambda^{\text{interior},\ast} = \frac{\E{ \sqrt{{g_1}{c_1}} + \sqrt{g_2 c_2}}^2}{ (b_0 - c_0)^2}, \]
which provides a useful initialization for numerical optimization. In this interior solution, the propensity $ \pi_1^\ast $ scales linearly with $ (b_0 - c_0) $, resulting in the asymptotic variance $ \mathrm{Var}^\pi_\infty(\hat{\xi}_n) $ decreasing with $ 1 / (b_0 - c_0)  $ as a function of the average invested budget $ b_0$. Third, in regions where the second-stage leverage dominates, so $ g_2 c_1 > g_1 c_2 $, the optimal propensity satisfies $ \pi^\ast_2 = 1 $, indicating that $ X_{M,r} $ will always be sampled jointly. Fourth, practical design planning in finite samples requires further optimization techniques across different budget levels, with sample sizes adjusted to maintain a fixed total cost \citep{mareis_2025}. Finally, designs in which the sampling decision may depend only on a restricted subset of variables, as in causal‑fairness settings where sensitive attributes must be protected \citep{pmlr-v80-kilbertus18a}, also fall within our framework. In such cases, the propensity expressions in Theorem \ref{theorem_optimization} are replaced by versions that integrate out the disallowed components, in the same way the unobserved error $\varepsilon_M$ is handled.

\subsection{Optimized Back-Door Estimation} \label{section_back_door}
Front-door estimation implicitly contains a back-door estimation problem on the subvector $ X_{C,t,M} $. By isolating the contribution of the optimized influence function $ \varphi_{\beta_{Mt}}^{\text{opt}} $ from Lemma \ref{lemma_observed_data_EIF}, one obtains an analogous problem for the back-door setting. The resulting propensity has the same structure as in Theorem \ref{theorem_optimization}, but depends only on the first-stage sampling decision.
\begin{corollary} \label{corollary_back_door}
Let $ c_0 > 0 $ be the cost of measuring $ X_{C,t} $ and let the measurable, positive cost function $ c_1 \in \mathcal{C}(\mathbb{R}^{d_C + 1}, \mathbb{R}_{> 0}) $ represent the cost of measuring $ X_{M} $ after $ X_{C,t} $ is already observed. For an average sample cost of $ b_0 \in \left( c_0, c_0 + \E{c_1(X_{C,t}) }\right] $, consider the optimization problem 
\newcommand{\argminexpr}{%
  \underset{
    \substack{
      \pi \in \mathcal{C}(\mathbb{R}^{d_C + 1}) \\
      \hat{\beta}_{Mt,n} \in \mathcal{E}_{\textrm{RAL}}({\beta}_{Mt}, \mathcal{M}_\pi)
    }
  }{\operatorname{argmin}}
  \; \mathrm{Var}_\infty(\hat{\beta}_{Mt,n};\pi)
}

\[
\begin{array}{c c l}
  \multirow{2}{*}{$\argminexpr$}
  & \qquad \text{s.t.} & 0 < \pi(X_{C,t}) \le 1 \text{ a.e.,} \\
  &             & \mathrm{E}[c_0 + (\pi c_1)(X_{C,t})] = b_0
\end{array}
\]
with the restriction that $ \varphi \in IF^\pi(\varphi_{\beta_{Mt}}^{F,\text{eff}}) $. If $ X_M $ is one-dimensional, the optimal propensity is 
\[
\pi^\ast = \min \left( 1, \sqrt{ \frac{ \varepsilon_t^2  \Var{\varepsilon_M}  }{\Var{\varepsilon_t}^{2} \lambda c_1 }} \right),
\]
where $ \lambda > 0 $ is chosen to satisfy the budget constraint. The corresponding estimator for $ \hat \beta_{Mt,n} $ coincides with the construction in Lemma \ref{lemma_linear_estimator}.
\end{corollary}

\section{Computational Results} \label{section_Computational_results}
\subsection{Simulation Study} \label{section_simulation_study}

For our simulations, we study the exemplary front-door model given in Appendix \ref{appendix_simulation_study}. The model features non-Gaussian errors on $X_{C,t,r} $ and Gaussian errors on $X_M$, multivariate confounder and mediator variables with $ d_C = 2$, $d_M = 3 $, a constant cost function $ c_2 $ and a non-constant cost function $ c_1(X_{C,t}) = 0.1 \cdot \Vert X_{C,t} \Vert_2 $. The study examines calibration, computational sensitivity as well as 
the dependence of the asymptotic variance in Lemma \ref{lemma_effect_EIF_var} and its optimization in Theorem \ref{theorem_optimization} on model parameters and nuisance components.

\begin{figure}
\includegraphics[width = 0.49\textwidth]{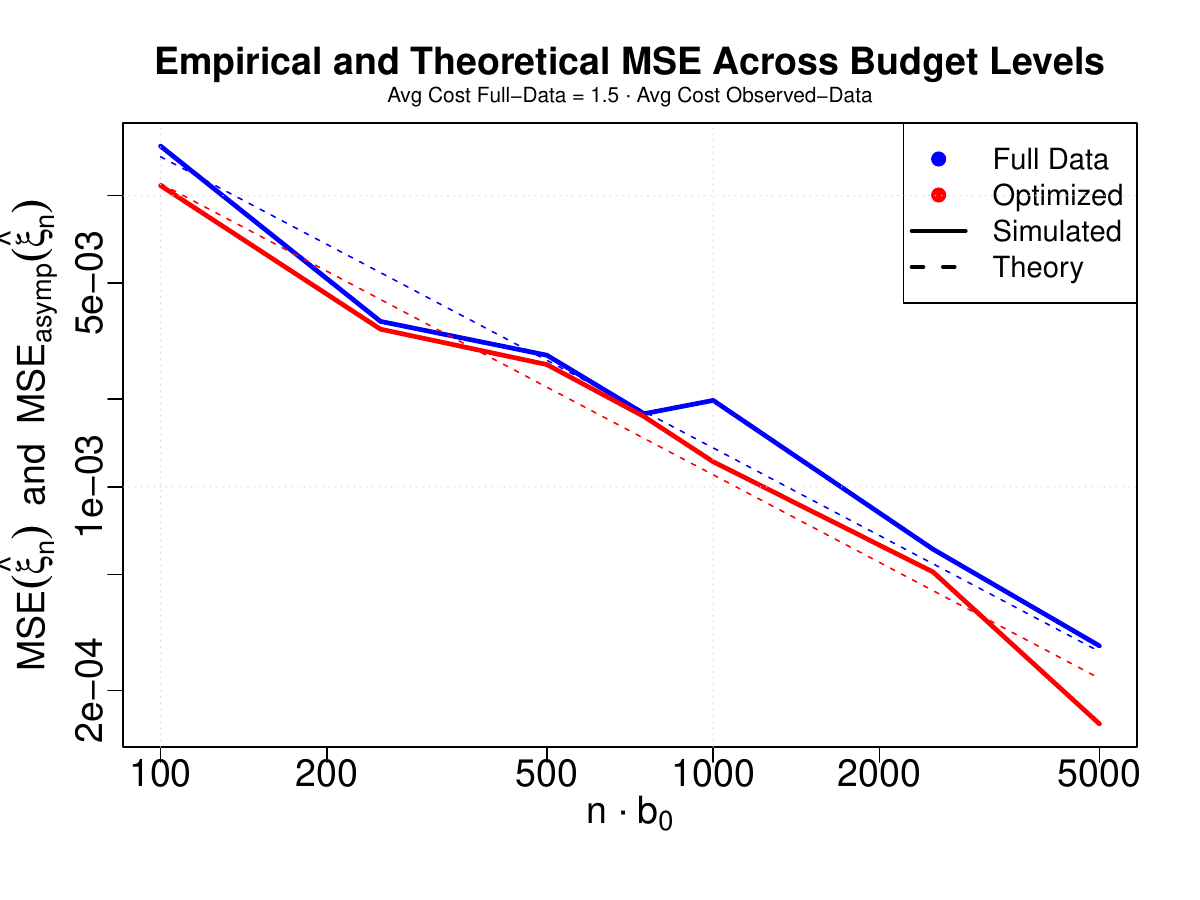}
\includegraphics[width = 0.49\textwidth]{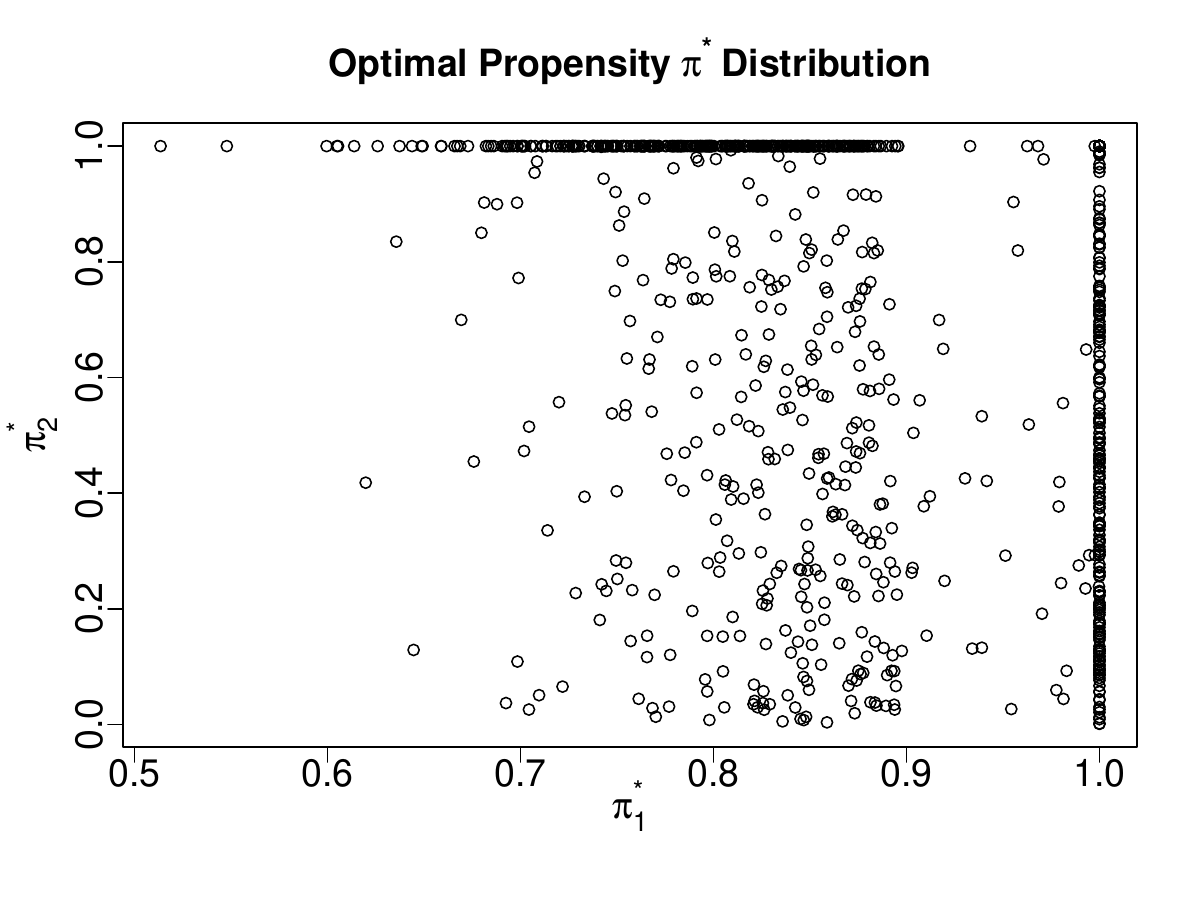}
\caption{Average empirical and theoretical MSE of the full-data and the optimized observed-data causal effect estimators across varying levels of budget $ n \cdot b_0 $ are presented in the left panel. In the right panel, the optimized propensity $ \pi^\ast $ distribution on $ 1,000$ samples is visualized, contrasting the full-data propensity $ \pi \equiv (1,1)^\top $.}
\label{figure_calibration}
\end{figure}
\begin{figure}
\center
\includegraphics[width = 0.235\textwidth]{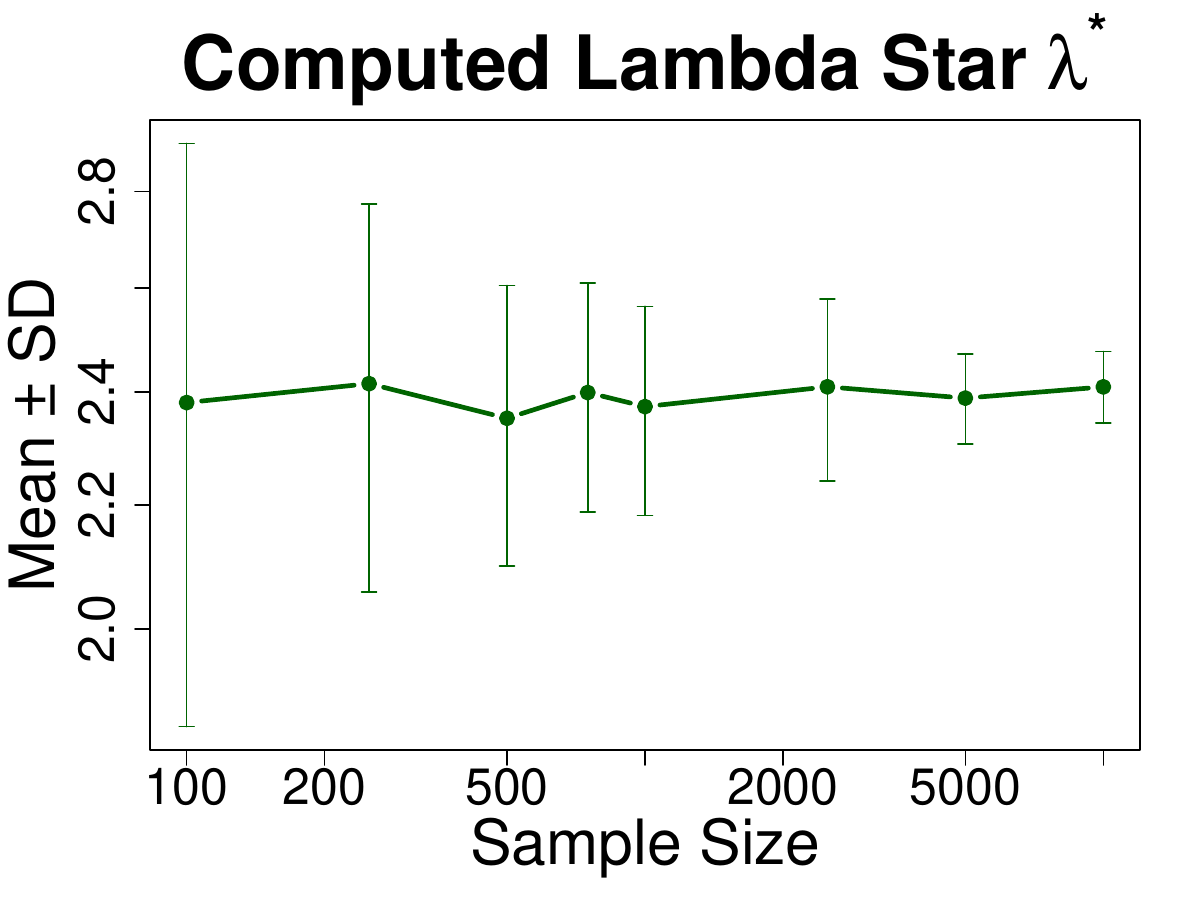}
\includegraphics[width = 0.235\textwidth]{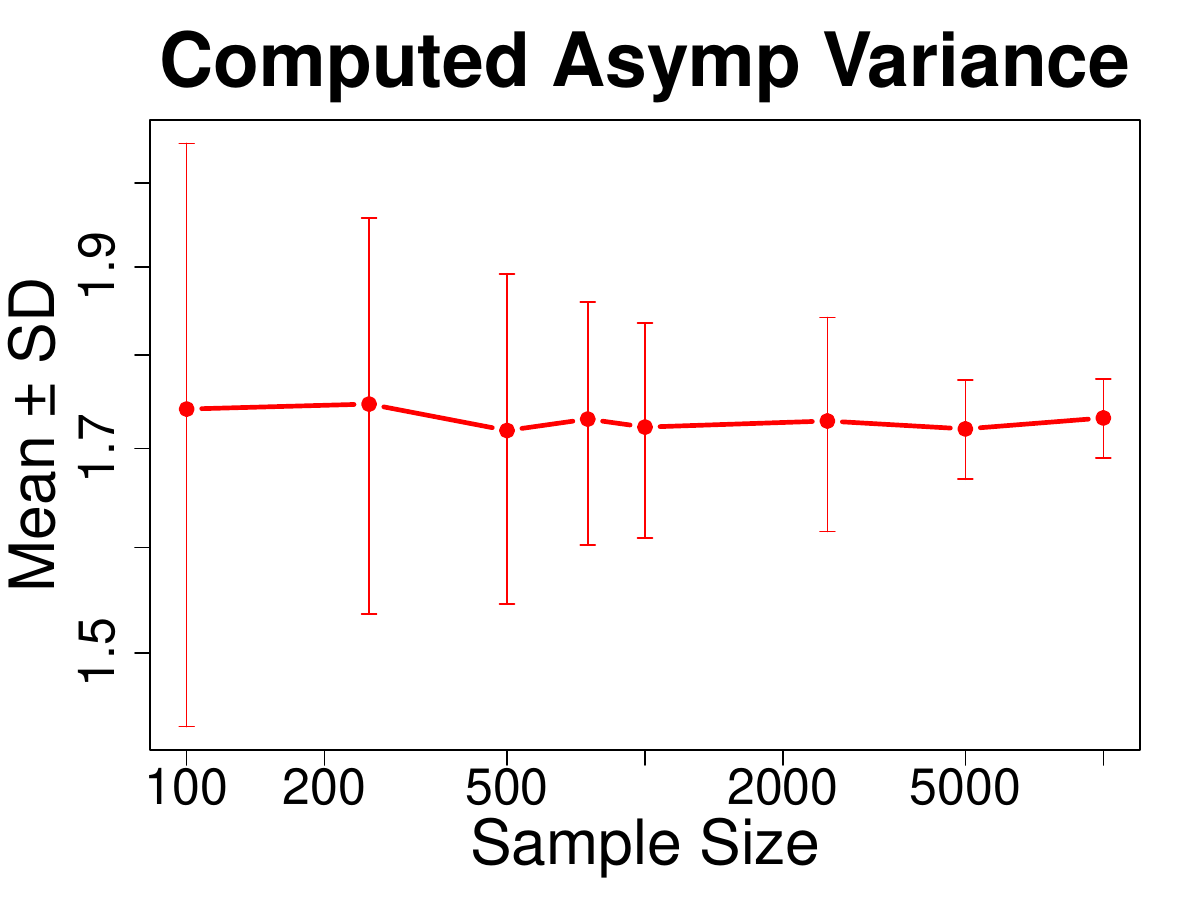} 
\includegraphics[width = 0.235\textwidth]{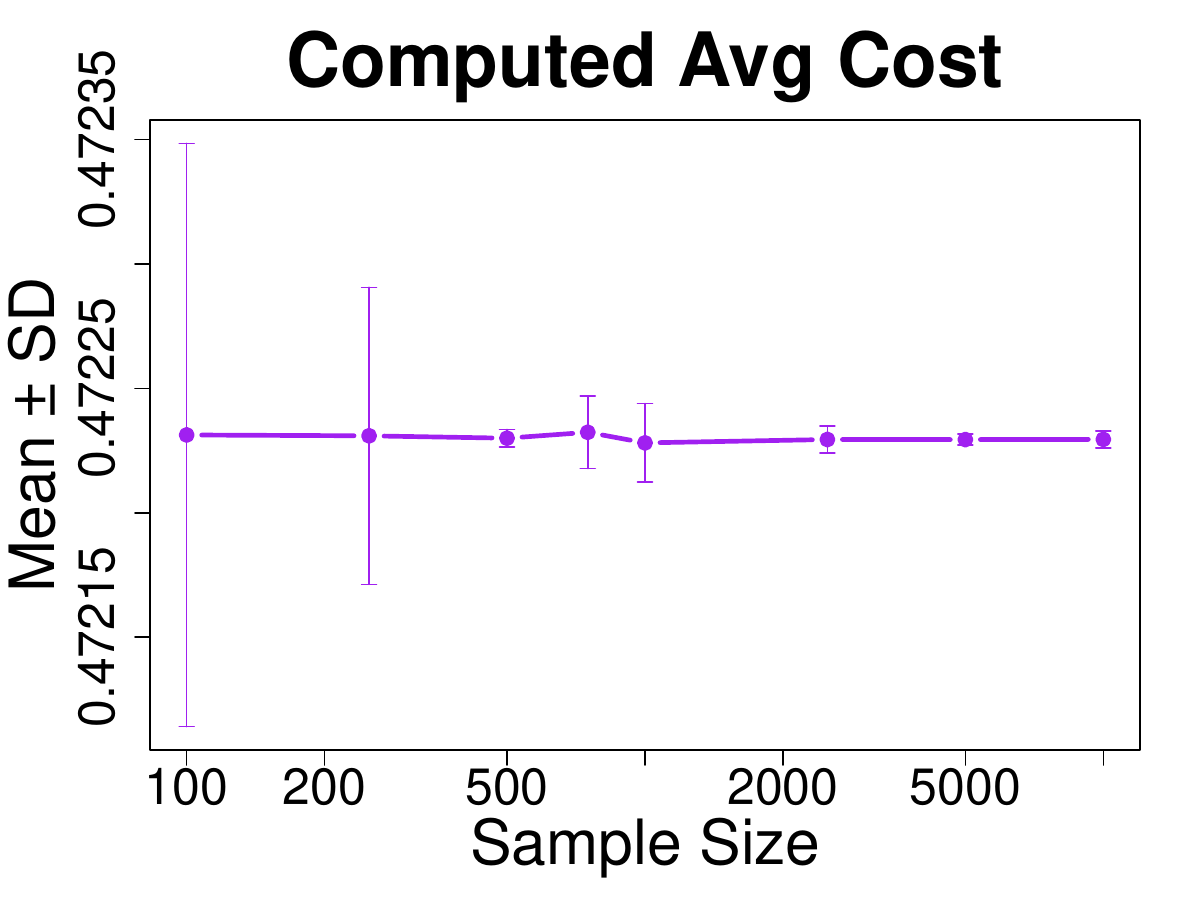}
\includegraphics[width = 0.235\textwidth]{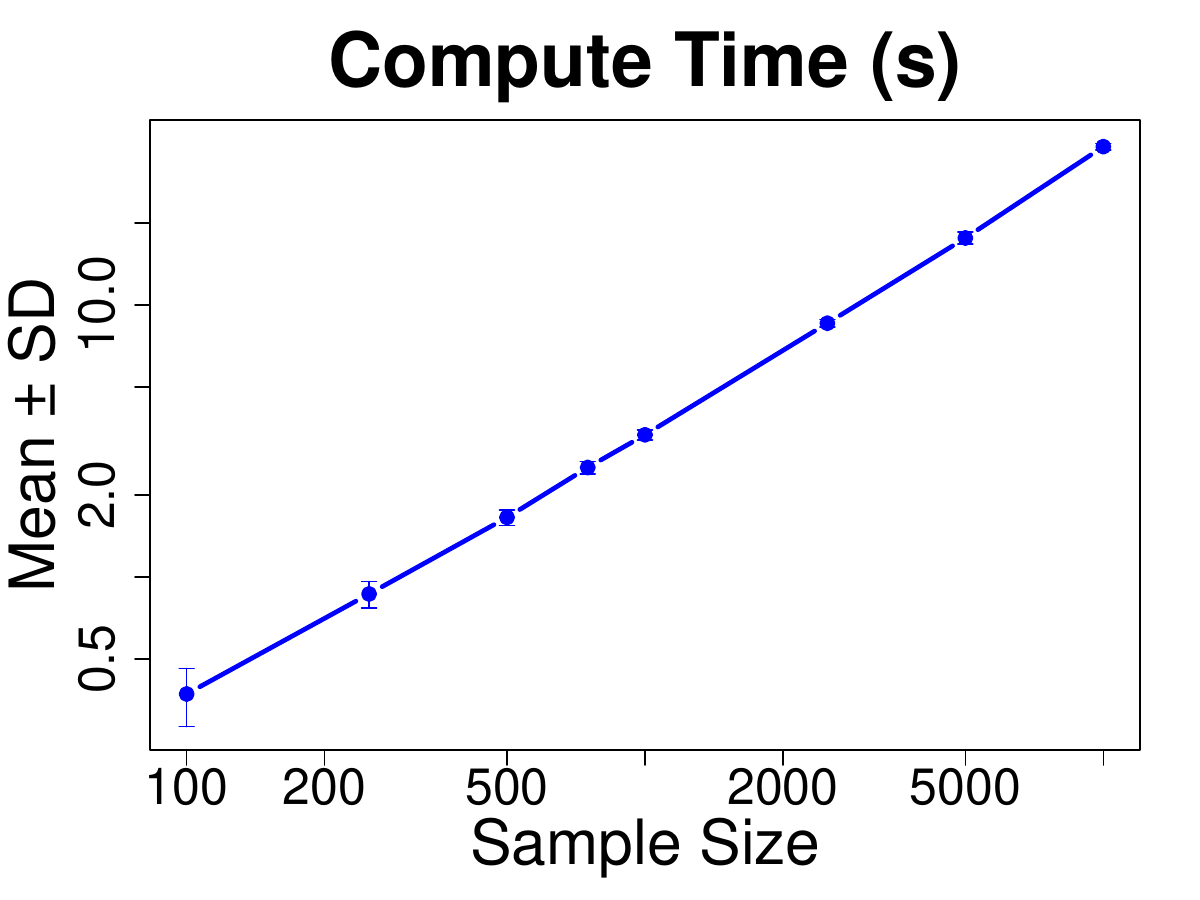}
\caption{Plots of the mean $ \pm 1 $SD of the optimized tuning parameter $ \lambda^\ast $, the computed asymptotic optimized variance $ \mathrm{Var}_\infty({\hat{\xi}_n};\pi^\ast)$, the computed average sample cost $ \E{c_0 + \pi_1^\ast c_1 + \pi_1^\ast \pi_2^\ast c_2} $ and compute time for sample sizes between $ 100 $ and $ 10,000$. }
\label{figure_computational_sensitivity}
\end{figure}
The first research question concerns calibration: How accurate does the theoretical asymptotic distribution of the causal effect estimator approximate the finite-sample mean-squared error (MSE), and what factors influence this reliability? 
Figure \ref{figure_calibration} compares in the left panel the theoretical asymptotic MSE with the average empirical MSE of both full-data and optimized observed-data estimators. The method is well calibrated even at modest sample sizes. There is a substantial increase in efficiency, or equivalently decrease in asymptotic MSE, when comparing the full-data design to the partial-data design at fixed budget levels, indicated by the x-axis. The optimal propensity function $ \pi^\ast $ is sensitive towards the input information and attains all four possible configurations $ \left((0,1), (0,1)\right)$, $ \left(1, (0,1)\right)$, $ \left((0,1),1\right) $, and $ \left(1,1\right) $. While Figure \ref{figure_calibration} reports  averages over $50$ replications, the uncertainty of a single asymptotic variance computation remains of interest. 

Figure \ref{figure_computational_sensitivity} displays the variability of the tuning parameter $ \lambda^\ast $, the asymptotic optimized variance $ \mathrm{Var}_\infty({\hat\xi_n};\pi^\ast)$, the realized average sample cost $ \E{c_0 + \pi_1^\ast c_1 + \pi^\ast_1 \pi^\ast_2 c_2} $, and the computational time. Although we find the average cost to be highly reliable, the variance in  optimized asymptotic exhibits non-negligible variability. In critical applications, we recommend quantifying the computational uncertainty, e.g., via bootstrapping, and protecting against it by considering conservative asymptotic variance estimates.

Next, we investigate how the marginal parameters and nuisance components affect the relative efficiency $  \mathrm{Var}_\infty({\hat\xi_n};{\pi^\ast}) / \mathrm{Var}_\infty({\hat\xi_n};1)$ and asymptotic variance, to identify scenarios where partial-measurement designs are particularly advantageous. A relative efficiency of $80\%$ means that, under the same budget, the optimized partial-measurements design achieves a $20\%$ reduction in asymptotic variance compared to full measurements. All experiments are summarized in Figure \ref{figure_sensitivity} and we restrict our discussion to the plots with major trends.
\begin{figure}
\includegraphics[width = 0.33\textwidth]{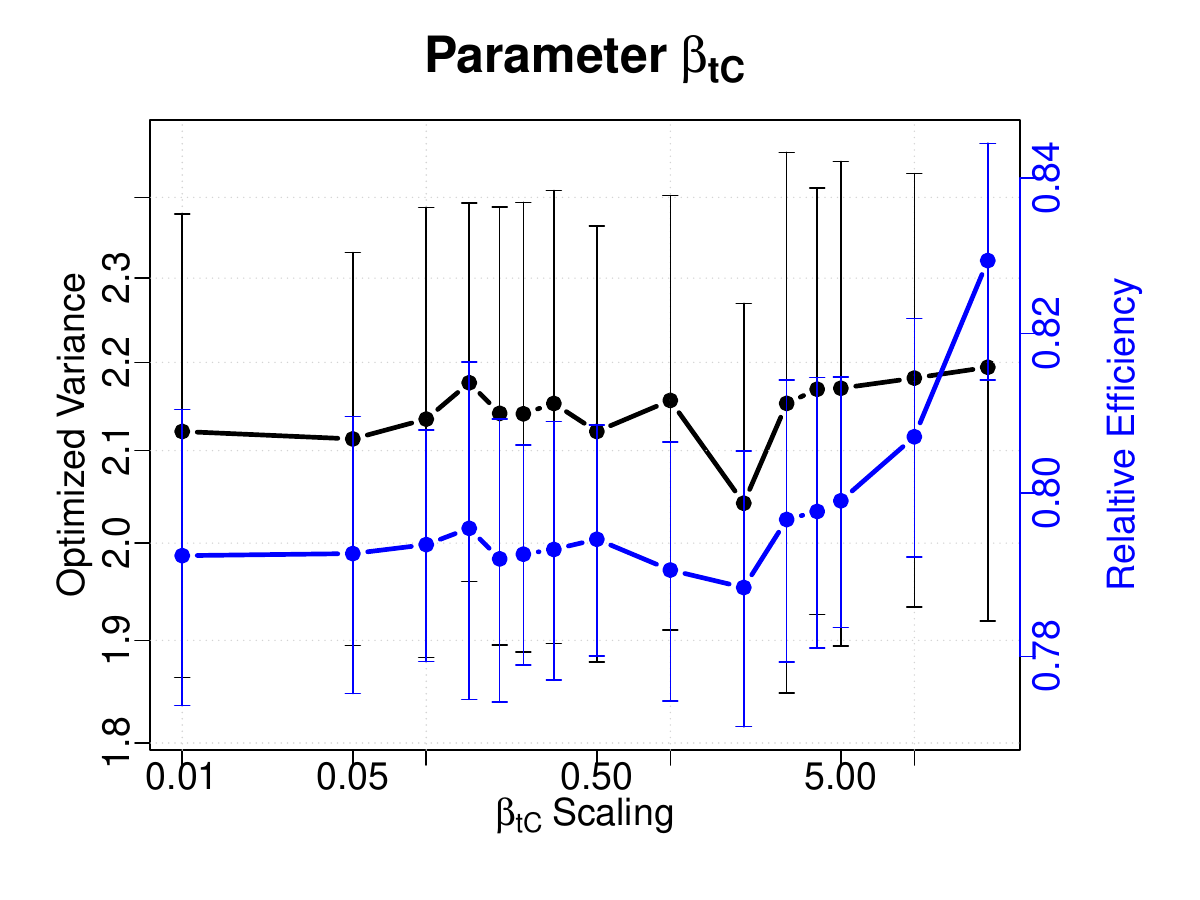} \includegraphics[width = 0.33\textwidth]{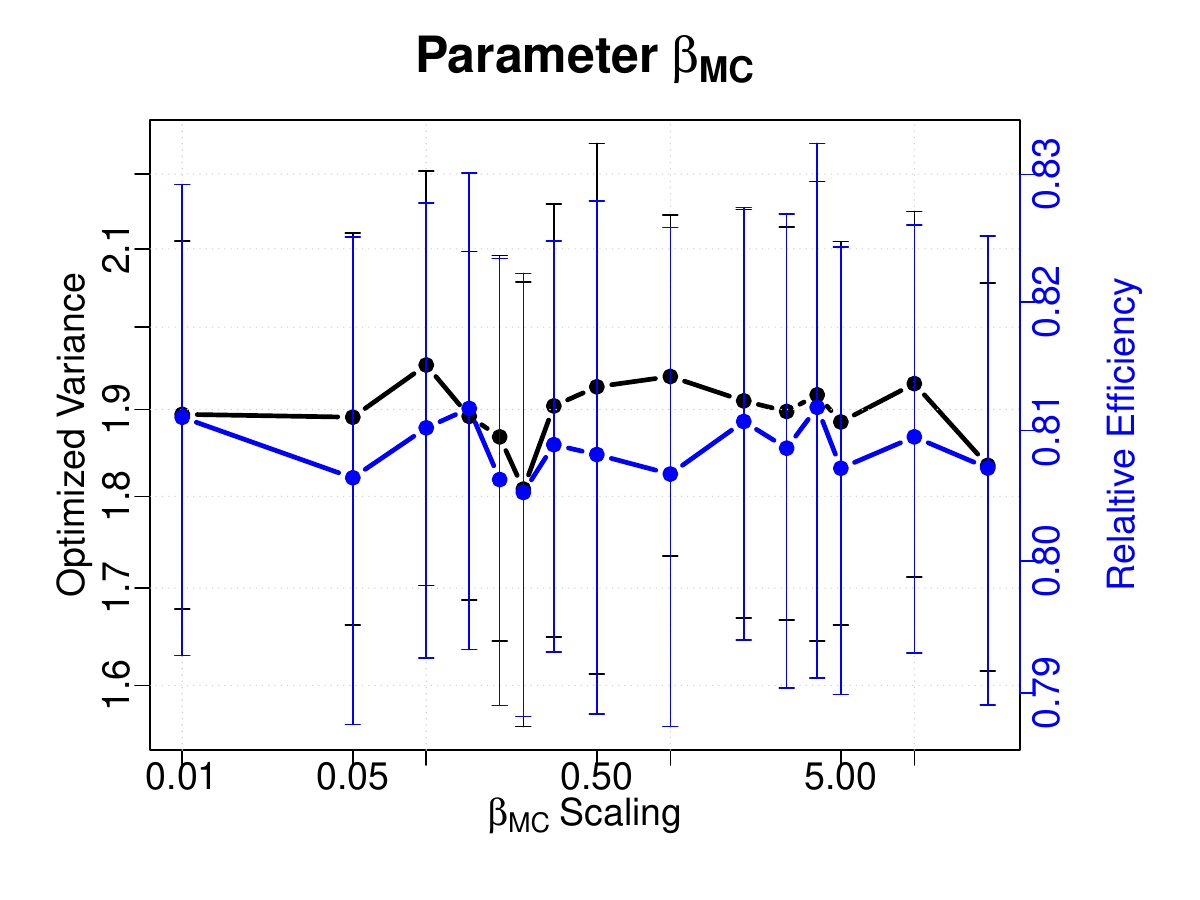} \includegraphics[width = 0.33\textwidth]{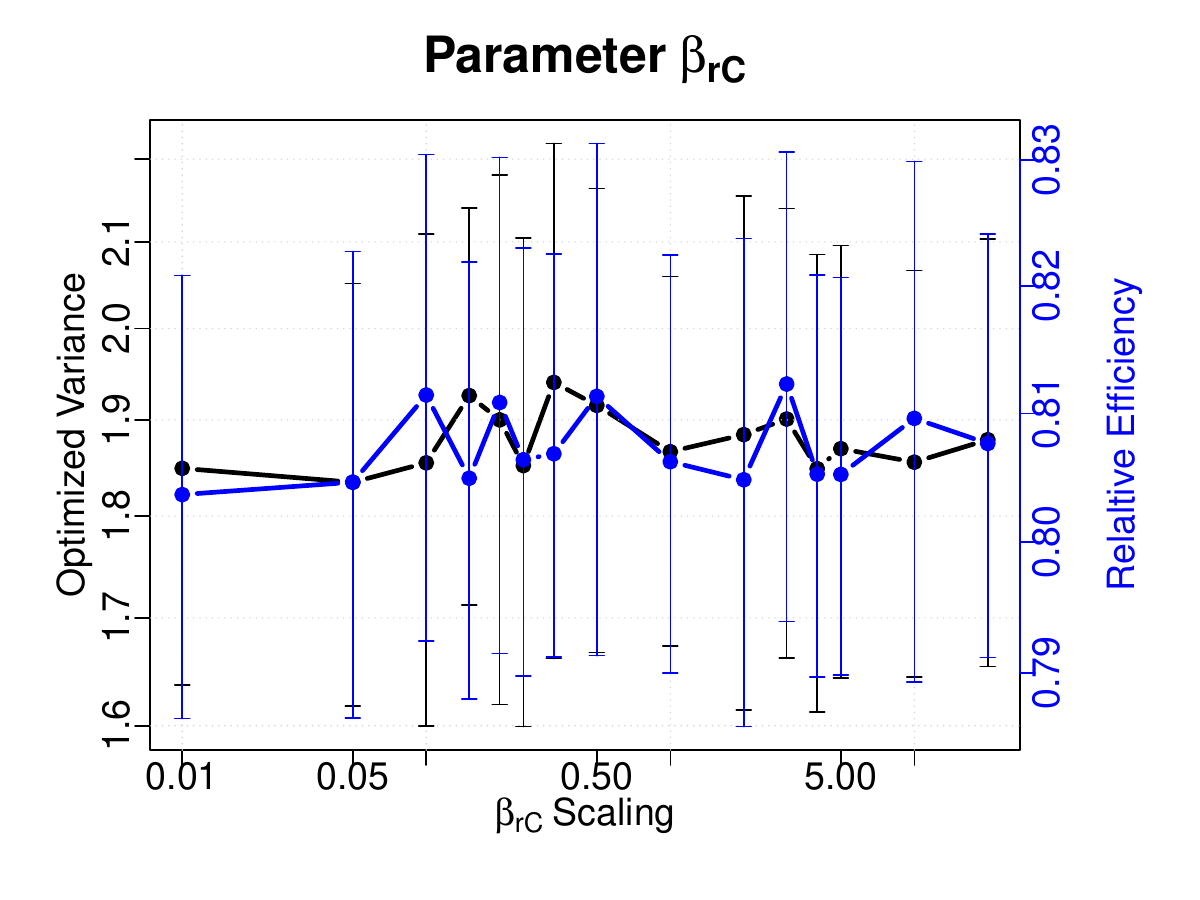}
\includegraphics[width = 0.33\textwidth]{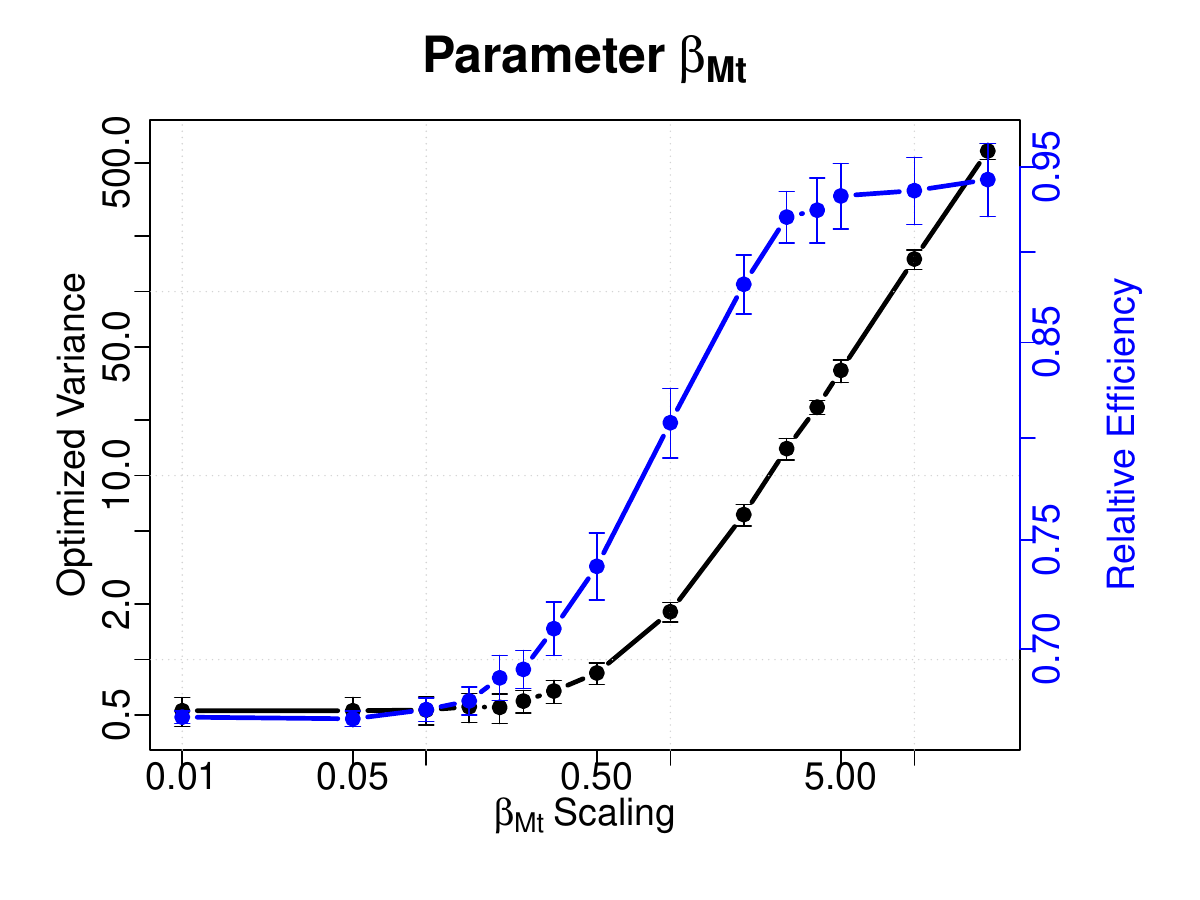}
\includegraphics[width = 0.33\textwidth]{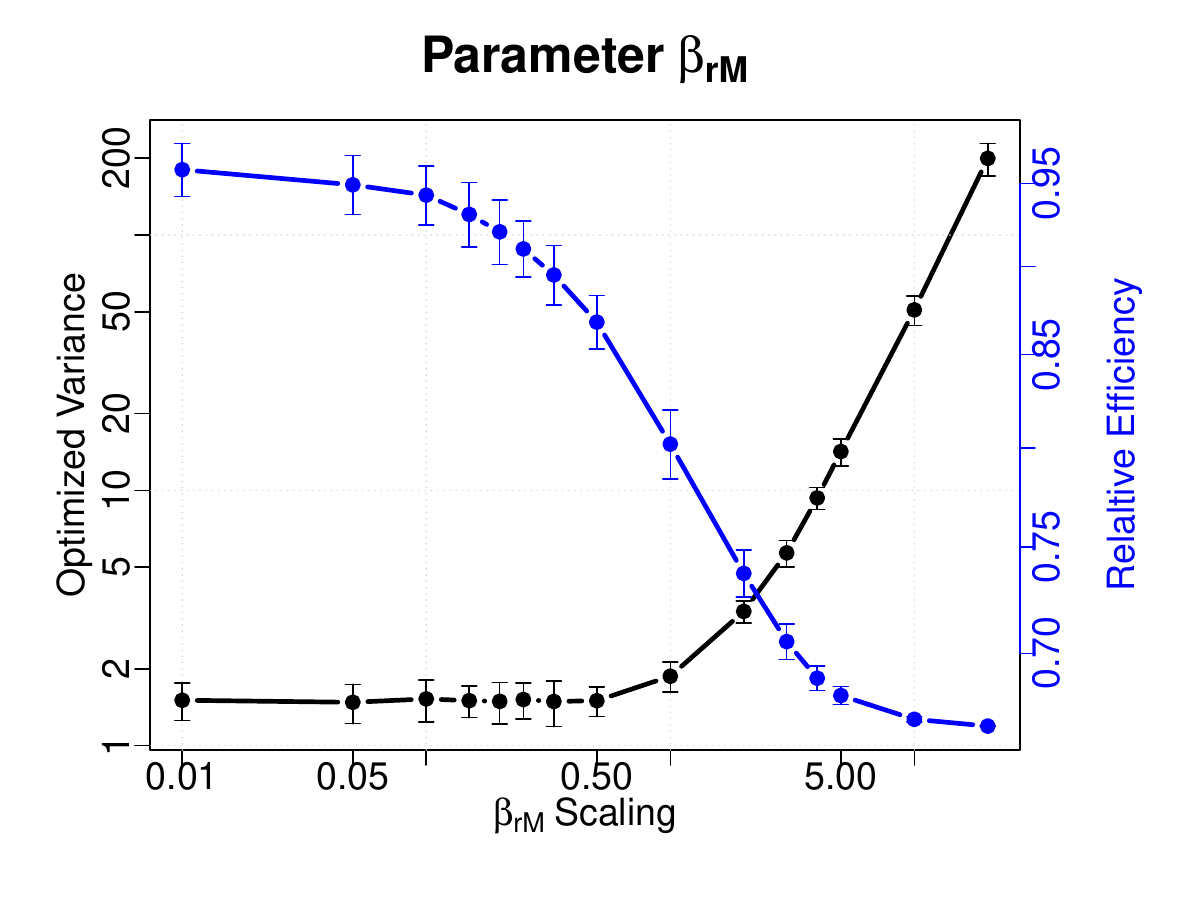}
\includegraphics[width = 0.33\textwidth]{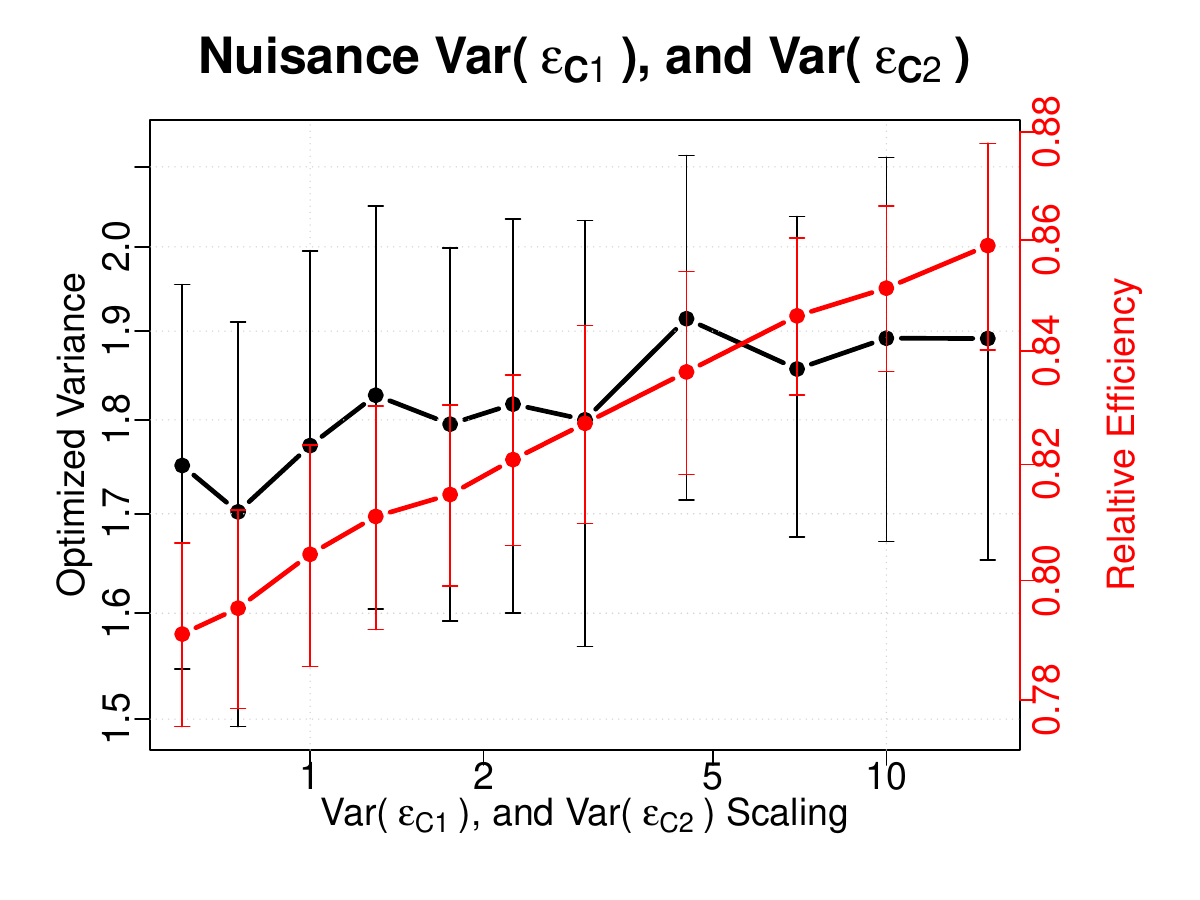}
\includegraphics[width = 0.33\textwidth]{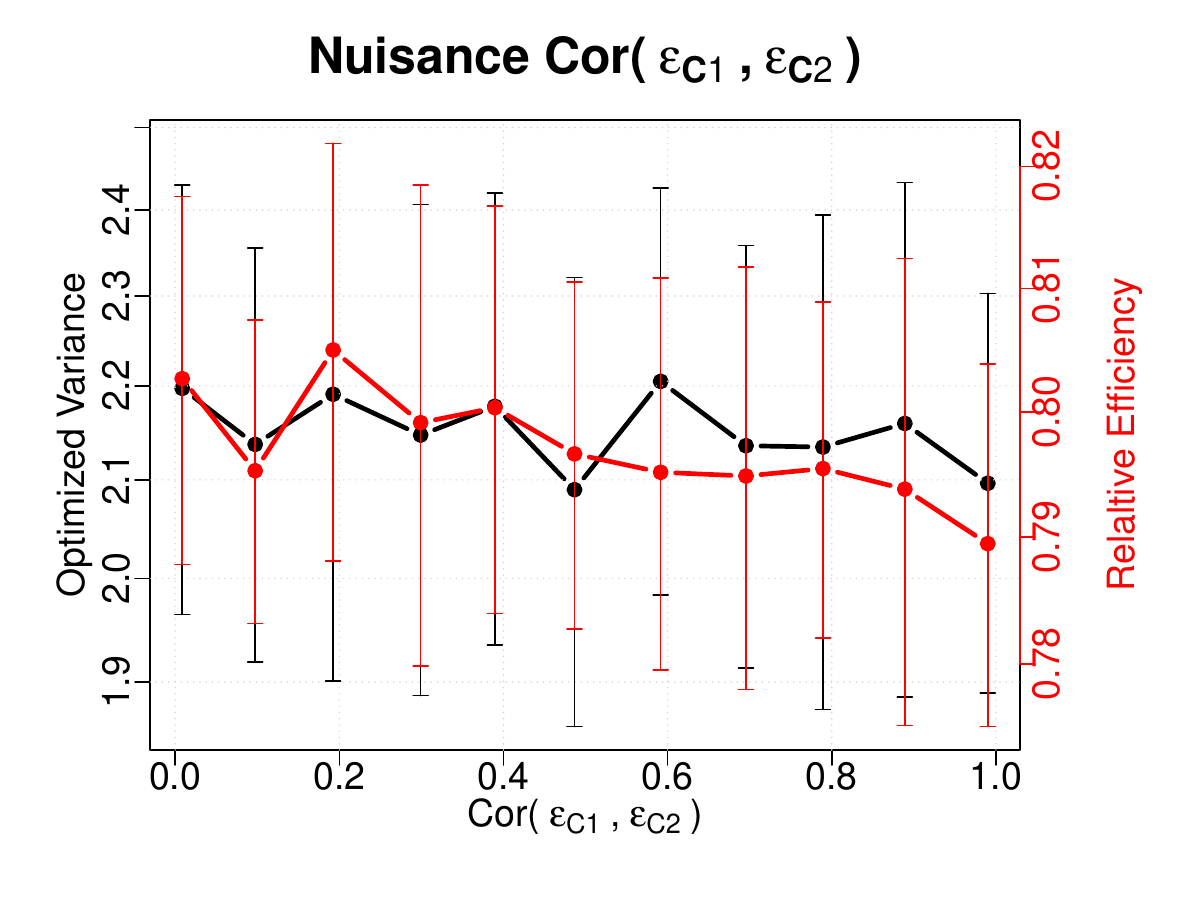}\includegraphics[width = 0.33\textwidth]{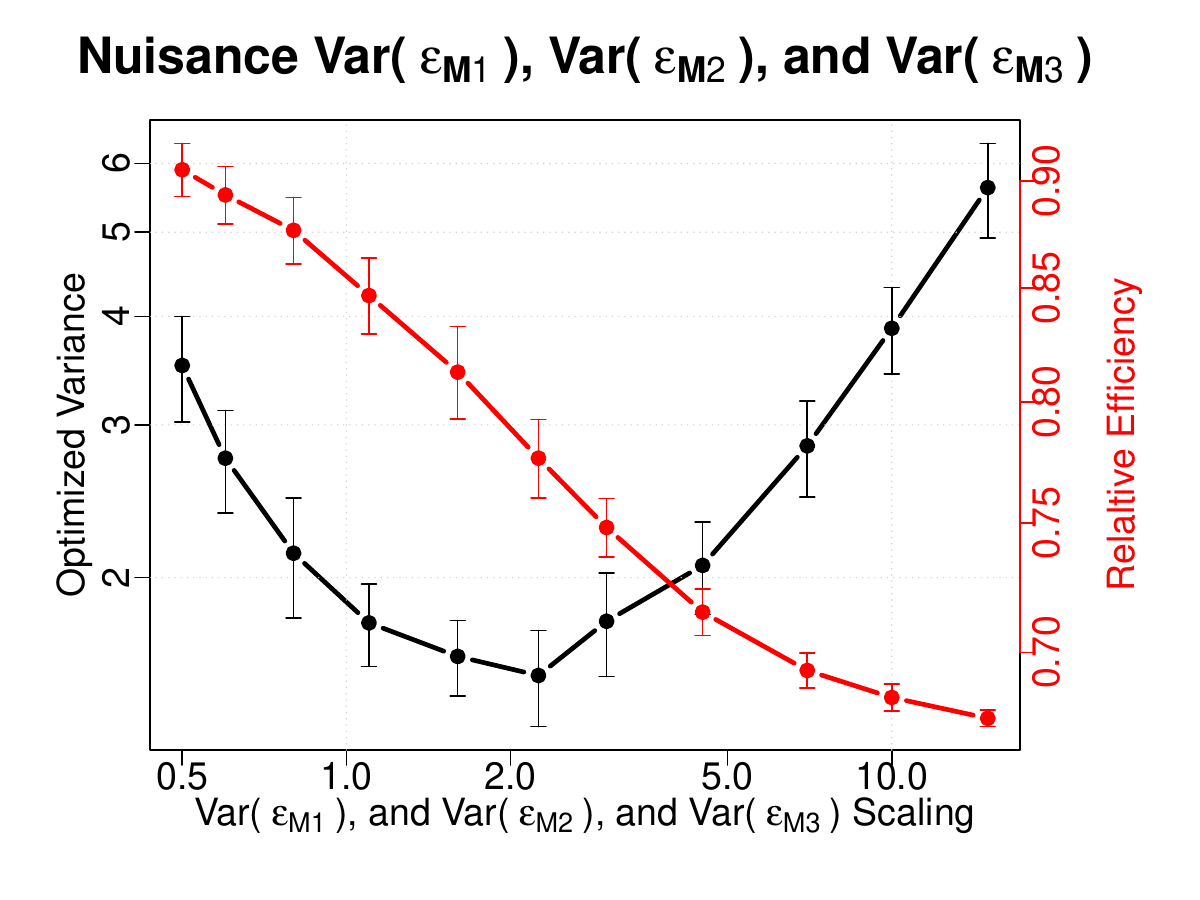}
\includegraphics[width = 0.33\textwidth]{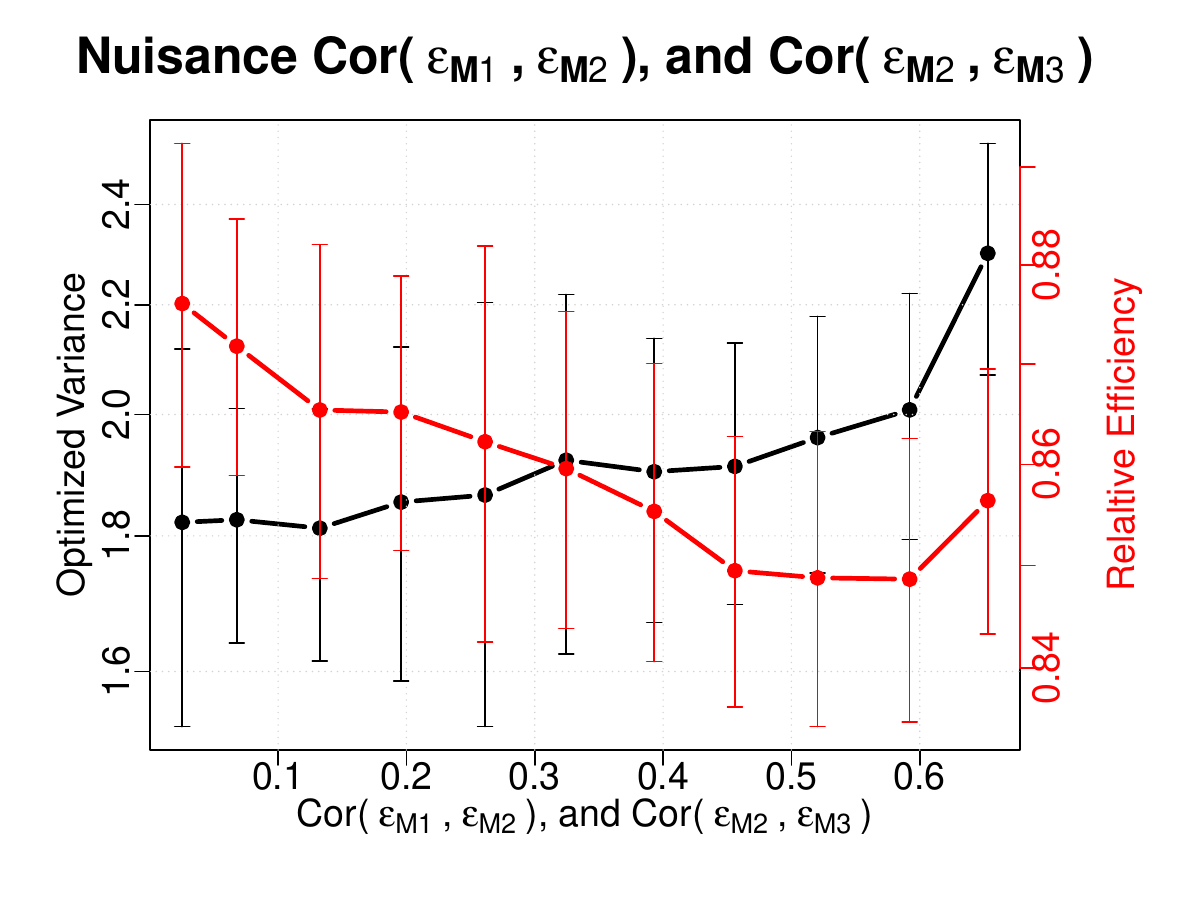}
\includegraphics[width = 0.33\textwidth]{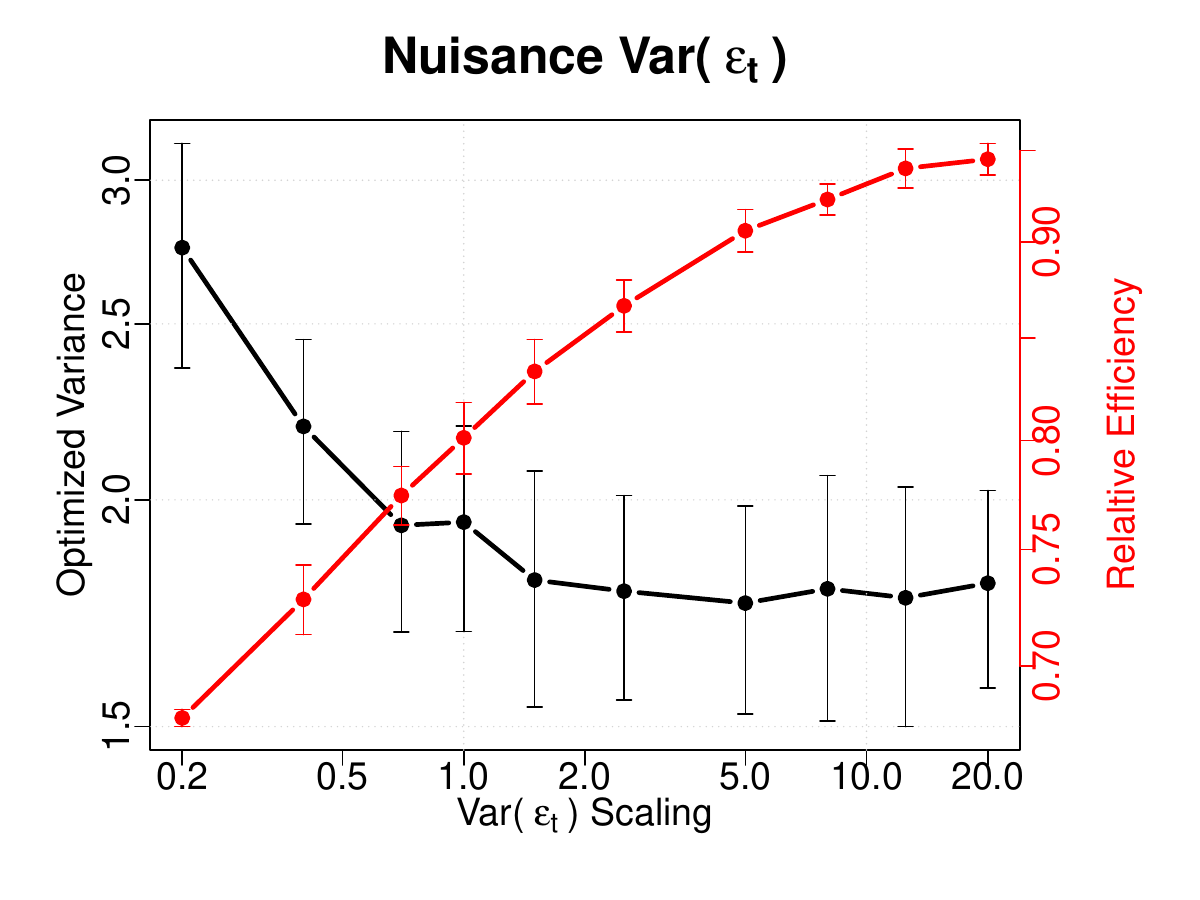}
\includegraphics[width = 0.33\textwidth]{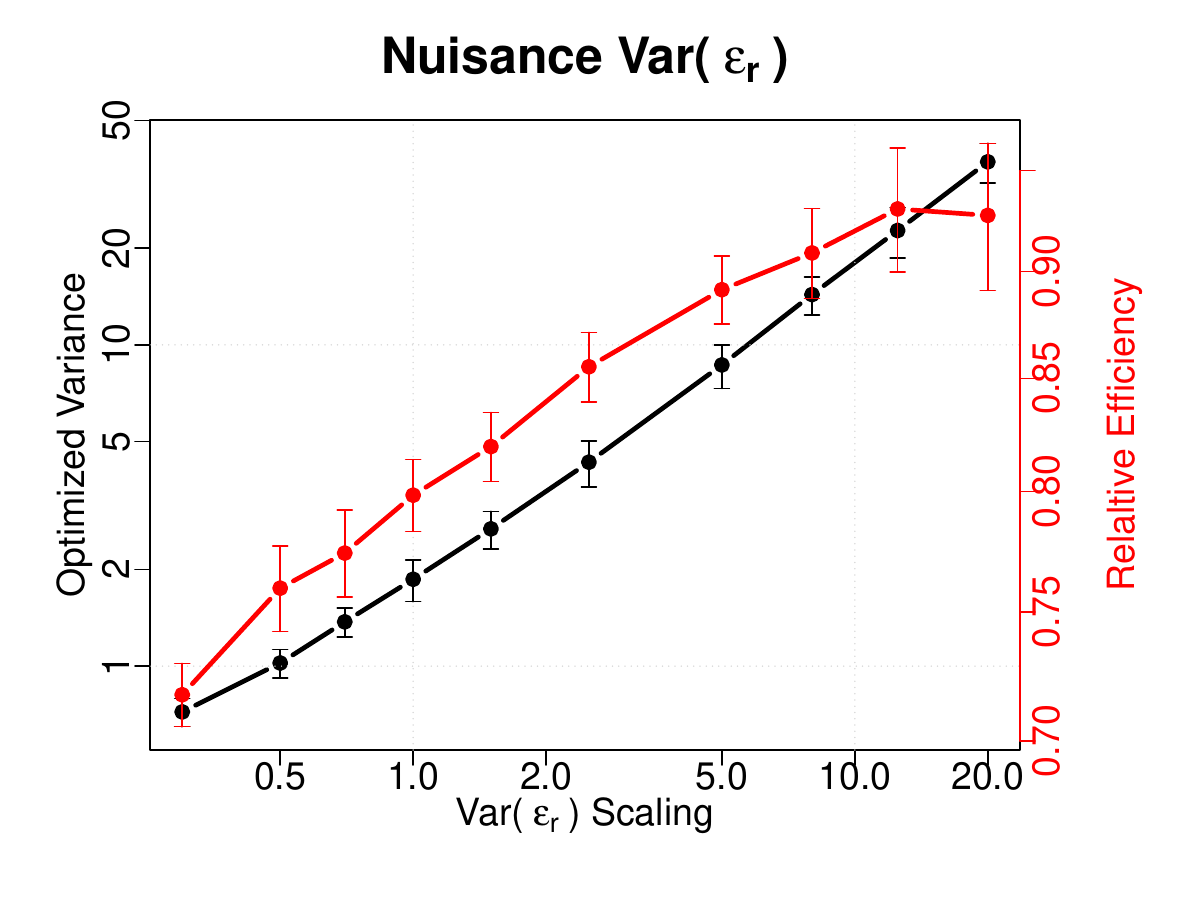}
\includegraphics[width = 0.33\textwidth]{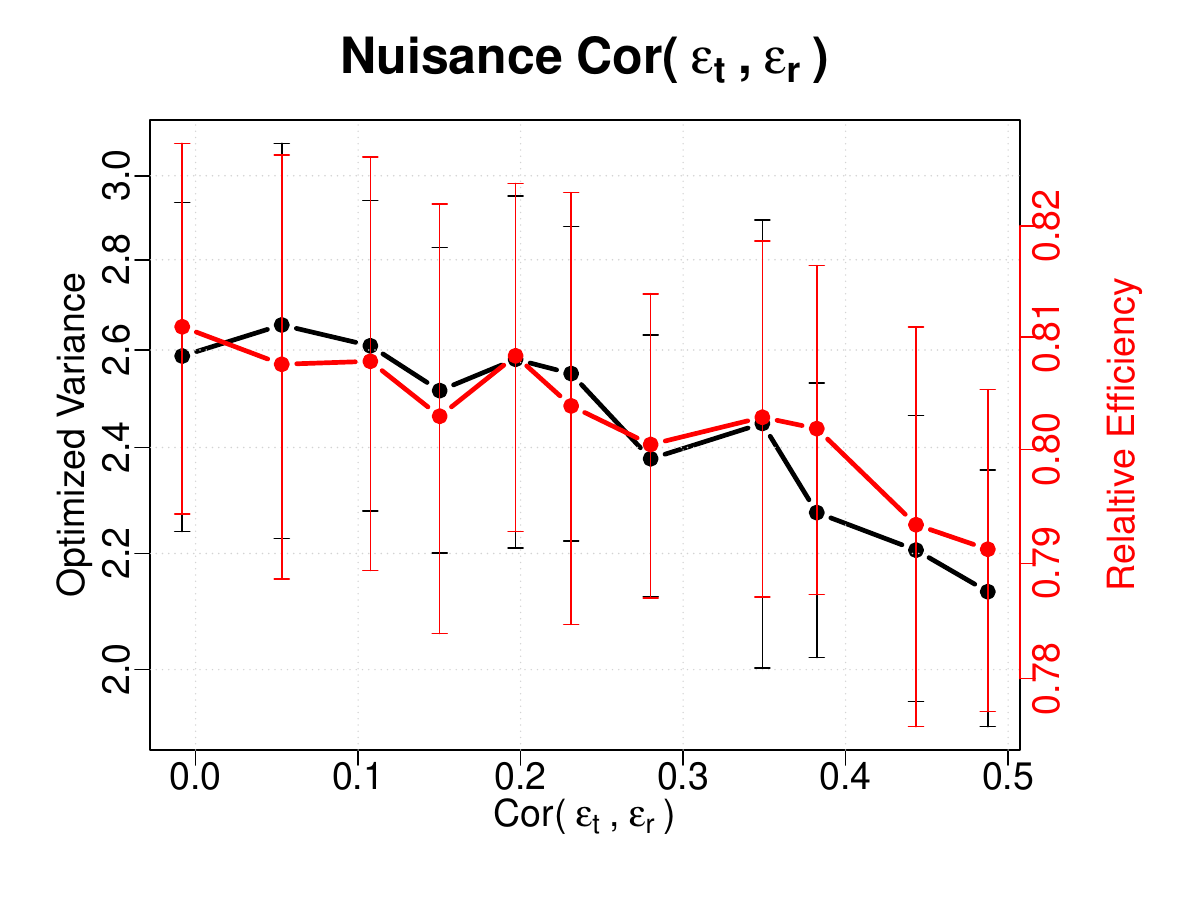}
\caption{The plots demonstrate the variation (mean $ \pm 1$SD) in asymptotic optimized variance $ \mathrm{Var}_\infty({\hat\xi^{\text{opt}}_n}; \pi^\ast)$ (black) and relative efficiency $ \mathrm{Var}_\infty({\hat\xi_n};\pi^\ast) / \mathrm{Var}_\infty({\hat\xi_n};1)$ (red / blue), and optimized variance by individual perturbations of the parameters and nuisance. For optimization, we specified the average cost of the observed-data estimator to be $ 2/3 $ of the full-data average cost, resulting in the relative efficiency scaling. All perturbations of parameters and nuisance are listed in Table \ref{table_alterations_experiment_parameter_nuisance_sensitivity}. The x-axis represents either the multiplicative scaling factor or the resulting (average) correlation(s). The y-axes differ in scale.} 
\label{figure_sensitivity}
\end{figure}
The general guideline is that partial-measurements designs are favorable when $ \beta_{Mt} $ is difficult to estimate or when $ \beta_{rM} $ is comparatively easy to learn. 
Such settings allow the design to exploit variance reduction through sample efficient estimation of $ \beta_{rM} $. This guideline is supported by the following findings. 
\begin{itemize}
    \item Decreasing the magnitude of parameter $ \Vert \beta_{Mt} \Vert $ (Figure \ref{figure_sensitivity} [2nd row, 1st col]), as well as reducing the variance of $ \varepsilon_t $ (Figure \ref{figure_sensitivity} [4, 1]) or increasing the variance of $ \varepsilon_M  $ (Figure \ref{figure_sensitivity} [3, 2]) make the estimation of $ \beta_{Mt} $ more challenging and increase the benefit of partial measurements. 
    \item Increasing the parameter $ \Vert \beta_{rM} \Vert $ (Figure \ref{figure_sensitivity} [2, 2]) or decreasing the variance of $ \varepsilon_r $ (Figure \ref{figure_sensitivity} [4, 2]) strengthens the signal along the path $ X_M \rightarrow X_r $, enabling $ \beta_{rM} $ to be estimated from relatively fewer samples and again enabling efficiency gains through partial measurements. 
    \item Increasing the variance of $ \varepsilon_C $ (Figure \ref{figure_sensitivity} [2, 3]) has a qualitatively different effect. It results in noisier back-door paths, making both $ \beta_{rM} $ and $ \beta_{Mt} $ harder to estimate and thereby degrading the performance of sensitive, leverage-based partial-measurements designs. 
    \item Finally, reducing the confounding between $ \varepsilon_t $ and $ \varepsilon_r $ (Figure \ref{figure_sensitivity} [4, 3]) transitions the system to an unconfounded setting, where full-data measurements of $ X_{t,r} $ suffice and the benefit of partial measurements diminishes. 
\end{itemize}

\subsection{Applications}
We demonstrate the advantages of partial-measurement designs on two real-world datasets and one semi-synthetic dataset drawn from biology, medicine and industry. 
\paragraph{Sachs dataset.}
The dataset of \cite{Sachs2005} ($ 853 $ samples; source: \href{doi.org/10.5281/zenodo.7679091}{10.5281/zenodo.7679091}, \cite{Mareis_Haug_Drton_2025}) records protein and phospholipid expression levels in human cells. The original study includes both observational and interventional data, yielding one of the few ground truth causal directed acyclic graphs (DAG) describing the influence between measurements. Within this graph, we determined two front-door configurations suitable for our analysis.
\paragraph{MIMIC-IV dataset.}
Medical professionals administer insulin to control elevated blood sugar levels of diabetes patients. In observational data, however, the effect of insulin is confounded by patient characteristics such as age and body mass index (BMI). To study this setting, we processed the electronic health dataset MIMIC-IV (\cite{MIMIC_Johnson}) of the Beth Israel Deaconess Medical Center, generating a dataset of size $ 999 $ samples on these four variables.
\paragraph{Industrial assembly dataset.}
The third dataset is generated using the Python package \texttt{causal\-Assembly} \citep{gobler24a}. The authors combined real-world assembly line measurements and expert knowledge to construct a graphical model. Their package provides a semi-synthetic sampling method in which the imposed graphical assumptions are satisfied by construction. Consequently, our back-door dataset on $3,000$ samples from \texttt{station1} meets the requirements.
\paragraph{}
Table \ref{table_results_application} reports the oversampling percentage and relative efficiency obtained by applying Theorem \ref{theorem_optimization}, respectively Corollary \ref{corollary_back_door}. Oversampling quantifies the required sample size under a partial measurement to match the total cost of a full-data design. For example, in the causalAssembly experiment, a partial-measurement dataset of $ 1.09 \cdot n $ yields a $ 5.3 \% $ variance reduction relative to a full-measurement dataset of size $ n $. Overall, the results confirm that exploiting parameter uncertainty through optimized experimental design can substantially reduce the estimation variance compared to the naive full-measurement design.

\begin{table} 
\centering
\begin{tabular}{lllllll}
\multirow{2}{*}{Dataset} & \multirow{2}{*}{$ X_C $} & \multirow{2}{*}{$ X_t $} & \multirow{2}{*}{$ X_M $} & \multirow{2}{*}{$ X_r $} & Oversampling & Relative \\ &&&&&Percentage&Efficiency \\  \hline
Sachs & PKA, PKC & Raf & Erk & Akt & $ 135 \% $ & $  0.7662$ \\ 
Sachs & PKA & Mek & Erk & Akt & $ 154 \% $ & $ 0.6802 $ \\
MIMIC & Age, BMI & Insulin & $ \Delta $ Glucose & & $ 115 \% $ &  $ 0.8987 $ \\
causalAssembly & S1mp3  & S1mp5 &  S1mp4 & & $ 109 \% $ &  $0.9470 $ \\ 
\end{tabular}
\caption{Results of Theorem \ref{theorem_optimization} andCorollary \ref{corollary_back_door} applied to three datasets. The cost functions are specified to be $ c_1 \equiv 1, c_2 \equiv 1 $ with base cost $ c_0 = 1 + d_C $ for Sachs and causalAssembly. Only for MIMIC we set $ c_0 = 1 $ as \textit{BMI} and \textit{Age} measurements are negligible.} 
\label{table_results_application}
\end{table}

\section{Conclusion}

The central contribution of this work is a unified framework that combines semiparametric efficiency theory and causal graphical modeling to address a problem that arises naturally in practice: how to allocate a limited measurement budget across study stages so as to minimize the variance of a causal effect estimator. The RAL estimators and optimized designs we derive show that substantial efficiency gains over full-measurement strategies are achievable without sacrificing inferential validity, a finding with direct relevance to experimental planning in medicine, biology, and industry.

The framework rests on three modeling choices whose relaxation we consider the most consequential directions for future work. First, the linearity assumption, which we regard as the strongest restriction imposed. It was essential for establishing a clean theoretical framework and deriving closed-form optimal designs, while offering a clear interpretation and serving as first-order approximations to more complex systems \citep{hastie2009, Pearl155170}. Appendix~\ref{appendix_misspecification} confirms that the optimized design retains its efficiency advantage under mild nonlinearity, though efficiency gains are expected to diminish as departures become more pronounced. Second, the present sampling schemes follow the intrinsic causal order. Designs that depart from this ordering would give rise to conceptually distinct augmented influence functions and may unlock additional efficiency gains, though the theoretical development would be substantially more involved. Here, we also see the possibility of sequential propensity updates offering additional finite-sample gains. And finally, the imposed causal front-door structure. Although the characterization of efficient observed-data influence functions beyond the augmentation class $IF^\pi(\varphi^{F,\text{eff}})$ remains open, extensions to identified causal graphs, including those visualized in Figure \ref{figure_identified_causal_graphs}, are feasible and promising. 

\input{Figure_Id_Graphs}

\acks{Leopold Mareis received funding from the German Research Foundation (DFG) under the Mathematical Research Data Initiative (project No. 460135501). Mathias Drton received funding from the European Research Council (ERC) under the European Union's Horizon 2020 research and innovation programme (grant
agreement No. 883818).}

\bibliography{bibliography}

\newpage
\appendix 

\section{Structure of the Results} \label{appendix_structure_of_results}
The following points summarize the logical structure of the results and clarify how the main lemmas and theorems relate to the influence‑function relations shown in Figure~\ref{figure_proof_logic}.
\begin{figure}
    \centering
    \input{rebuttal/Figure_proof_logic}
    \caption{Logical structure of the full‑data and observed‑data influence functions, augmentation classes, and the feasible set used in the optimization.}
    \label{figure_proof_logic}
\end{figure}
\begin{itemize}
    \item It is known that the affine space $ \varphi_0^F + \mathscr{T}^\perp$ characterizes all full-data influence functions, given in Figure~\ref{figure_proof_logic} by the {\color{violet} violet set}. 
    \item Theorem \ref{thm_full_data_EIF} identifies the full-data efficient influence function, shown in Figure \ref{figure_proof_logic} as $ \bm{\varphi^{F,\mathrm{eff}}_2}$. 
    \item Lemma \ref{lemma_tsiatis} provides a construction of augmented influence functions for each observed-data model $ \mathcal{M}_{(\pi_1,\pi_2)}$. Each full-data influence function $ \varphi^F$ generates a class $ IF^\pi(\varphi^F)$ of augmented observed-data influence functions. The disjoint union of all generated classes spans the set of augmented observed-data influence functions in $ \mathcal{M}_{(\pi_1,\pi_2)}$. Each arrow in Figure \ref{figure_proof_logic} illustrates the generation of such a class.
    \item Lemma \ref{lemma_observed_data_EIF} determines the most efficient representative within each augmentation class. These are the bold observed-data influence functions in Figure \ref{figure_proof_logic}, for example $ \bm{\varphi_{2.2}^{\pi^1}}$ in $ {\color{orange} IF^{\pi^1}(\varphi^{F,\mathrm{eff}}_2)}$ or $ \bm{\varphi_{1.3}^{\pi^2}}$ in $ {\color{red} IF^{\pi^2}(\varphi^{F}_1)}$.
    \item The main Problem \ref{problem_main} seeks the minimal-asymptotic-variance RAL estimator $\hat{\xi}_n$ over all models $\mathcal{M}_{\pi_1, \pi_2}$. Since the asymptotic variance equals $\mathrm{E}[\varphi^2]$, the problem is equivalent to minimizing the length of the observed-data influence function over all models. The feasible set corresponds to the entire second row of Figure \ref{figure_proof_logic}. If such an observed-data influence function were identified, the experimenter could then sample data from the corresponding optimal model and use the associated estimator for cost-efficient causal effect estimation.
    \item For technical reasons, however, we must restrict the search for observed-data influence functions to a smaller feasible set, namely to the augmentation classes generated by the efficient full-data influence function. In Figure \ref{figure_proof_logic}, this is depicted in the third row. 
    \item Lemma \ref{lemma_effect_EIF_var} computes, for each model $ \mathcal{M}_{(\pi_1,\pi_2)}$, the length of the most efficient observed-data influence function within $ IF^{\pi}(\varphi^{F,\text{eff}})$. This equals the asymptotic variance of the corresponding estimator. 
    \item Theorem \ref{theorem_optimization} then identifies the optimal influence function $ \varphi^{\textrm{opt}}$ within this restricted feasible set, based on Lemma 10. As discussed in the main text, $ \varphi^{\mathrm{opt}}$ might not be globally efficient. This would be the case if $ \mathrm{E}[(\bm{\varphi^{\pi^1}_{1.2}})^2] < \mathrm{E}[(\varphi^{\textrm{opt}})^2] \leq \mathrm{E}[(\bm{\varphi^{\pi^1}_{2.2}})^2]$ in Figure \ref{figure_proof_logic}.
    \item Although not necessarily observed-data efficient across all models, $\varphi^{\textrm{opt}}$ remains compelling. It improves (a) on all influence functions in its own augmentation class $IF^{\pi^\ast}(\varphi^{F,\text{eff}})$ and (b) on all other models' augmentation class of the observed-data efficient influence function. If $\varphi^{\textrm{opt}} = \bm{\varphi^{\pi^2}_{2.3}}$ in Figure \ref{figure_proof_logic}, this implies (a) $ \mathrm{E}[(\bm{\varphi^{\pi^2}_{2.3}})^2] <  \mathrm{E}[(\bm{\varphi^{\pi^2}_{2.1}})^2] $ and (b) $ \mathrm{E}[(\bm{\varphi^{\pi^2}_{2.3}})^2] <  \mathrm{E}[(\bm{\varphi^{\pi^1}_{2.2}})^2] $.
\end{itemize}

\section{Proofs} \label{appendix_proofs}

\subsection{Proof of Theorem \ref{thm_full_data_EIF}: Efficient Influence Function in Model \texorpdfstring{$\mathcal{M}_1$}{M\_1}}
\restatement{Theorem \ref{thm_full_data_EIF}}{\TextTheoremFullDataEIF}
We need an auxiliary lemma for the proof. It derives the orthogonal complement of the nuisance tangent space $ \Lambda_\varepsilon^\perp $. This characterization is needed to construct the efficient score and, after normalization, the efficient influence function.

\begin{lemma} 
Let $ \varepsilon_r^\perp = \varepsilon_r - \E{\varepsilon_r  \, \middle| \, \varepsilon_t}$. The orthogonal complement to the $ \varepsilon $ nuisance tangent space is given by 
\begin{align} \label{formula_nuisance_orthogonal_compelement}
\Lambda_{\varepsilon}^{\perp} := \left\{ 
(g_C^\top,  g_t,  g_M^\top, g_r(\varepsilon_t))
\begin{pmatrix}
\varepsilon_C \\ 
\varepsilon_t \\
\varepsilon_M \\
\varepsilon_r^\perp
\end{pmatrix}
\,\middle|\, 
\begin{array}{l}
	g_C \in \mathbb{R}^{d_C}, \\
	g_t \in \mathbb{R}, \\
	g_M \in \mathbb{R}^{d_M}, \\
	g_r \in L^2(\mathbb{R})
\end{array} \right\}.
\end{align}
\end{lemma}

\begin{proof}
Using the factorization of the error term in Equation \eqref{formula_varepsilon_factorization}, the nuisance tangent space $ \Lambda_\varepsilon $ decomposes as $ \Lambda_{\varepsilon_C} \oplus \Lambda_{\varepsilon_{t,r}} \oplus \Lambda_{\varepsilon_M} $. Only the component $ \Lambda_{\varepsilon_{t,r}} $ needs further refinement and we instead use the representation $ \Lambda_{\varepsilon_{t,r}} = \Lambda_{\varepsilon_{t}} \oplus \Lambda_{\varepsilon_{r | t}}$. Let $ \zeta_j: \mathbb{R}^d \to \mathbb{R}^{d_j}, \quad \zeta: x \mapsto x_j $ denote the projection onto the components for all $ j \in \{ C, t, M, r\} $. Similar to \citep{Tsiatis_Semiparametric}, Thm. 4.6 and Thm 4.7, the nuisance tangent spaces for $ j \in \{ C, t, M\} $ are
\begin{align*}
\Lambda_{\varepsilon_j} &= \left\{ h \in L^2_0(\mathbb{R}^d) \,\middle|\, \exists h^\prime \in L^2_0(\mathbb{R}^{d_j}) : h =  h^\prime \circ \zeta_j \wedge \E{h \varepsilon_j} = 0 \right\} \quad \text{ , and} \\
\Lambda_{\varepsilon_{r|t}} &= \left\{ h \in L^2(\mathbb{R}^d) \,\middle|\, \exists h^\prime \in L^2(\mathbb{R}^2) : h =  h^\prime \circ \zeta_{t,r} \wedge \forall \varepsilon_t:  \E{h \,\middle|\, \varepsilon_t} = 0 ,  \E{h \varepsilon_r  \,\middle|\, \varepsilon_t} = 0 \right\}.
\end{align*}
Any concatenation in the restriction, e.g., $ h^\prime \circ \zeta_C$, must be understood that the function $ h $ is only allowed to depend on the entries in respective component, here $ C$. The tangent spaces $ ( \Lambda_{\varepsilon_C}, \Lambda_{\varepsilon_t}, \Lambda_{\varepsilon_M}) $ are pairwise orthogonal as they act on different components and so is $ \Lambda_{\varepsilon_{r|t}} $ to $ ( \Lambda_{\varepsilon_C}, \Lambda_{\varepsilon_M}) $. For $ h_1 \in \Lambda_{\varepsilon_t} $ and $ h_2 \in \Lambda_{\varepsilon_{r|t}} $ we find orthogonality via $
\E{h_1^\top h_2} = \E{\E{h_1^\top h_2 \mid \varepsilon_t}} = \E{h_1^\top \E{ h_2 \mid \varepsilon_t}} = 0 $ as $ h_1 $ is $ \sigma(\varepsilon_t) $-measurable and $ \E{ h_2 \mid \varepsilon_t} = 0$. Next, we show that the orthogonal space to $ \Lambda_{\varepsilon}$ is given by Equation \eqref{formula_nuisance_orthogonal_compelement}. Since the nuisance tangent space $ \Lambda_{\varepsilon}$ is a direct sum, its orthogonal complement is the intersection of the orthogonal complements of its components. For $ j \in \{ C, t, M \} $, consider the proposed orthogonal spaces
\[ 
\Lambda_{\varepsilon_j, \text{prop}}^\perp =  \left\{h \in L^2(\mathbb{R}^{d}) \, \middle| \, \exists g_j \in \mathbb{R}^{d_j} \text{ with } \E{h| \varepsilon_j} =  g_j^\top \varepsilon_j  \right\}. 
\] 
To see $ \Lambda_{\varepsilon_j, \text{prop}}^{\perp} \subseteq \Lambda_{\varepsilon_j}^{\perp} $, note that for any $ a \in \Lambda_{\varepsilon_j} $ and $ h \in \Lambda_{\varepsilon_j, \text{prop}}^{\perp} $, we find $ \E{ah} = \E{a \E{ h \, \middle| \, \varepsilon_j} } = g_j^\top \E{a  \varepsilon_j} = 0$. To show $ \Lambda_{\varepsilon_j, \text{prop}}^{\perp} \supseteq \Lambda_{\varepsilon_j}^{\perp} $, let $ h \in L_0^2(\mathbb{R}^d) $ be arbitrary. We define the function $ h_1 := ( h - \E{h\, \middle| \, \varepsilon_j}) + \E{h \varepsilon_j}^\top \Var{\varepsilon_j}^{-1}  \varepsilon_j \in \Lambda_{\varepsilon_j, \text{prop}}^{\perp}$ and argue that $ h - h_1 \in \Lambda_{\varepsilon_j} $. The function $ h - h_1 = \E{h\, \middle| \, \varepsilon_j} - \E{h \varepsilon_j}^\top \Var{\varepsilon_j}^{-1}  \varepsilon_j $ is by construction $ \sigma(\varepsilon_j)$-measurable and as also $ \E{ h - h_1 } = 0 $ as well as $ \E{ ( h - h_1) \varepsilon_j } = 0 $, we conclude $  \Lambda_{\varepsilon_j, \text{prop}}^{\perp} = \Lambda_{\varepsilon_j}^{\perp} $. Analogously, we define for the conditional component
\[
\Lambda_{\varepsilon_{r | t}, \text{prop}}^\perp =  \left\{h \in L^2(\mathbb{R}^{d}) \, \middle| \, \exists g_0, g_r \in L^2(\mathbb{R}) : \E{h| \varepsilon_{t,r}} =  g_0(\varepsilon_t) + g_r(\varepsilon_t)^\top \varepsilon_r^\perp \right\}. 
\] 
To see $ \Lambda_{\varepsilon_{r|t}, \text{prop}}^{\perp} \subseteq \Lambda_{\varepsilon_{r|t}}^{\perp} $, note that for any $ a \in \Lambda_{\varepsilon_{r|t}} $ and $ h \in \Lambda_{\varepsilon_{r|t}, \text{prop}}^{\perp} $, we find 
\begin{align*}
\E{ah} & = \E{a \E{ h \, \middle| \, \varepsilon_{t,r}} } = \E{a g_0(\varepsilon_t)} +  \E{a  g_r(\varepsilon_t)^\top \varepsilon_r^\perp} \\ &= \E{\E{a| \varepsilon_t} g_0(\varepsilon_t)} +  \E{a  g_r(\varepsilon_t)^\top \varepsilon_r}  - \E{a  g_r(\varepsilon_t)^\top \E{\varepsilon_r | \varepsilon_t}} \\ & = 0 + \E{ g_r(\varepsilon_t)^\top\E{a \varepsilon_r| \varepsilon_t}   }  - \E{\E{a| \varepsilon_t}  g_r(\varepsilon_t)^\top \E{\varepsilon_r | \varepsilon_t}} = 0 - 0.
\end{align*}
To show $ \Lambda_{\varepsilon_{r|t}, \text{prop}}^{\perp} \supseteq \Lambda_{\varepsilon_{r|t}}^{\perp} $, let $ h \in L_0^2(\mathbb{R}^d) $ be arbitrary. We define the function 
\[ 
h_1 := ( h - \E{h\, \middle| \, \varepsilon_{t,r}}) + \E{h  \, \middle| \, \varepsilon_t}  +  \E{h \varepsilon_r^\perp \, \middle| \, \varepsilon_t}^\top \Var{\varepsilon_r \, \middle| \, \varepsilon_t}^{-1} \varepsilon_r^\perp \in \Lambda_{\varepsilon_{r|t}, \text{prop}}^{\perp} 
\]
 and argue that $ h - h_1 \in \Lambda_{\varepsilon_{r|t}} $. The function $ h - h_1 $ is by construction $ \sigma(\varepsilon_{t,r})$-measurable. Using $ \E{  \varepsilon_r (\varepsilon_r^\perp)^\top \, \middle| \, \varepsilon_t} =  \Var{  \varepsilon_r \, \middle| \, \varepsilon_t} $  we find 
\begin{align*}
\E{ h - h_1  \, \middle| \, \varepsilon_t } & = \E{\E{h\, \middle| \, \varepsilon_{t,r}} - \E{h  \, \middle| \, \varepsilon_t}  -  \E{h \varepsilon_r^\perp \, \middle| \, \varepsilon_t}^\top \Var{\varepsilon_r \, \middle| \, \varepsilon_t}^{-1} \varepsilon_r^\perp  \, \middle| \, \varepsilon_t} \\ 
& = 0 - \E{h \varepsilon_r^\perp \, \middle| \, \varepsilon_t}^\top \Var{\varepsilon_r \, \middle| \, \varepsilon_t}^{-1} \E{ \varepsilon_r^\perp \, \middle| \, \varepsilon_t } = 0 \quad \text{, and} \\ 
\E{ ( h - h_1) \varepsilon_r \, \middle| \, \varepsilon_t} 
&= \E{\E{h\, \middle| \, \varepsilon_{t,r}} \varepsilon_r - \E{h\, \middle| \, \varepsilon_{t}} \varepsilon_r -   \E{h \varepsilon_r^\perp \, \middle| \, \varepsilon_t}^\top \Var{\varepsilon_r \, \middle| \, \varepsilon_t}^{-1} \varepsilon_r^\perp \varepsilon_r \, \middle| \, \varepsilon_t}\\
& = \E{h( \varepsilon_r - \E{\varepsilon_r | \varepsilon_t}) \varepsilon_r  \, \middle| \, \varepsilon_t}   - \E{  \varepsilon_r (\varepsilon_r^\perp)^\top \, \middle| \, \varepsilon_t}   \Var{\varepsilon_r \, \middle| \, \varepsilon_t}^{-1} \E{h \varepsilon_r^\perp \, \middle| \, \varepsilon_t}  =0.
\end{align*}
Intersecting the orthogonal complements $ \Lambda_{\varepsilon_C}^\perp, \Lambda_{\varepsilon_{t}}^\perp,  \Lambda_{\varepsilon_M}^\perp, \Lambda_{\varepsilon_{t | r}}^\perp $ yields the expression for $ \Lambda_\varepsilon^\perp $ in the lemma. The linear $ \varepsilon_t $ term in $ \Lambda_{\varepsilon_{t}}^\perp $ eliminates the functional term $ g_0(\varepsilon_t) $ from the conditional component $ \Lambda_{\varepsilon_{t | r}}^\perp $, completing the characterization. 
\end{proof}
\begin{proof}[Theorem \ref{thm_full_data_EIF}]
To obtain the efficient score for parameter $ \beta_{ij} $, $ i,j \in [d]$, in $ \mathcal{M}_1 $, we project the full-data score $ S_{\beta_{ij}}^F = (\partial / \partial \beta_{ij}) \log (p_\varepsilon) $ onto the nuisance tangent space $ \Lambda_{\varepsilon} $. The projection takes the form
\[
\Pi( S_{\beta_{ij}}^F | \Lambda_{\varepsilon}  ) = 
\begin{pmatrix}
\E{S_{\beta_{ij}}  \varepsilon_C} \\ 
\E{S_{\beta_{ij}}  \varepsilon_t} \\ 
\E{S_{\beta_{ij}}  \varepsilon_M} \\
\E{S_{\beta_{ij}} \varepsilon_r^\perp | \varepsilon_t}
\end{pmatrix}^\top
\E{
\begin{pmatrix}
\varepsilon_C \varepsilon_C^\top & 0 & 0 & 0 \\ 
0 & \varepsilon_t^2 & 0 & 0 \\ 
0 & 0 & \varepsilon_M \varepsilon_M^\top & 0 \\ 
0 & 0 & 0 & \varepsilon_r^2 \\ 
\end{pmatrix}
}^{-1} 
\begin{pmatrix}
\varepsilon_C \\
\varepsilon_t \\
\varepsilon_M \\ 
\varepsilon_r^\perp \\ 
\end{pmatrix}
\]
The expectations $ \E{S_{\beta_{ij}}  \varepsilon_k} $ follow from the centered the error terms $ \varepsilon_C, \varepsilon_t, \varepsilon_M, \varepsilon_r^\perp $ and the structural equation. Let $ i \neq r $ and let $ p_{\varepsilon_i, 0} $ denote the true density of~$ \varepsilon_i $. We have 
\[ 
\int \left(x_i -  \beta_{ij} x_j  - \beta_{i(-j)} x_{-j} \right) p_{\varepsilon_i, 0}\left(x_i -  \beta_{ij} x_j  - \beta_{i(-j)} x_{-j} \right) \; \mathrm{d}x_i = 0 
\]
for all $ \beta \in \mathcal{B} $ and $ x_{-i} = x_{1, \dots, i-1, i+1, \dots, d} \in \mathbb{R}^{d - 1} $. Note that $ \beta_{i, i} = 0 $.
Taking the derivative w.r.t. $ \beta_{ij} $ at the truth $ \beta_{ij, 0 } $ and interchanging integration and differentiation yields for $ \beta_{ij, 0 } \neq 0 $ 
\begin{align} \label{equation_expectation_score_error}
\int (- x_j) p_{\varepsilon_i, 0}(\varepsilon_i) \; \mathrm{d} \varepsilon_i + \int \varepsilon_i S_{\beta_{ij}}(\varepsilon_i) p_{\varepsilon_i, 0}(\varepsilon_i) \; \mathrm{d} \varepsilon_i = 0 \Leftrightarrow  X_j = \mathrm{E}_{\varepsilon_i}[ \varepsilon_i S_{\beta_{ij}}].
\end{align}
If $ \beta_{ij, 0 } $ is however equal to $ 0 $, this derivation gives $ \mathrm{E}_{\varepsilon_i}[ \varepsilon_i S_{\beta_{ij}}] =0 $. Similarly, we get $ \mathrm{E}_{\varepsilon_i}[ \varepsilon_i S_{\beta_{kj}}] =0 $ for all $ k \neq i $ as $ \varepsilon_i $ expressed in $ X $ is not depending on $ \beta_{kj} $. Under $ i = r$, we can rely for all $ \varepsilon_t $ on 
\[
\int \left(x_r -  \beta_{r,} x  - \E{\varepsilon_r | \varepsilon_t} \right) p_{\varepsilon_r | \varepsilon_t, 0}\left(x_r -  \beta_{r,} x, \varepsilon_t \right) \; \mathrm{d}x_r = 0 .
\]
Now we differentiate $ \partial / \partial \beta_{rj} $ to evaluate at $ \beta_{rj, 0 } \neq 0 $ and get $ X_j = \mathrm{E}_{\varepsilon_r | \varepsilon_t}[ \varepsilon_r^\perp S_{\beta_{rj}} | \varepsilon_t ] $ from
\begin{align*}
\int (- x_j) p_{\varepsilon_r | \varepsilon_t, 0}(\varepsilon_r \mid \varepsilon_t) \; \mathrm{d} \varepsilon_r + \int \varepsilon_r^\perp S^{r | t}_{\beta_{rj}}(\varepsilon_r \mid \varepsilon_t) p_{\varepsilon_r | \varepsilon_t, 0}(\varepsilon_r \mid \varepsilon_t) \; \mathrm{d} \varepsilon_r & = 0 \\ 
\Leftrightarrow \int (- x_j) p_{\varepsilon_r | \varepsilon_t, 0}(\varepsilon_r \mid \varepsilon_t) \; \mathrm{d} \varepsilon_r + \int \varepsilon_r^\perp S_{\beta_{rj}}(\varepsilon) p_{\varepsilon_r | \varepsilon_t, 0}(\varepsilon_r \mid \varepsilon_t) \; \mathrm{d} \varepsilon_r &= 0 
\end{align*}
as the score decomposed due to independencies into $ S(\varepsilon) = S^{r | t}(\varepsilon_r | \varepsilon_t) + \sum_{j \in \{ C, t, M\}} S^j(\varepsilon_j) $. The remaining scores $ S_{\beta_{kj}} $ all have zero conditional expectation with the residual response $ \varepsilon_r^\perp $ given $ \varepsilon_t $. The efficient influence function $ \varphi_{\beta} $ is obtained by orthogonalizing the projected score $ \Pi( S_{\beta_{ij}}^F | \Lambda_{\varepsilon}  ) $ with respect to the parameter space $ \beta $ and normalizing by the inverse Fisher information matrix $ \E{S^{F, \text{eff}} (S^{F, \text{eff}})^\top}^{-1}$. This yields for $ \beta_{tC} $ the full-data efficient influence function
\begin{align*}
\varphi_{\beta_{tC}}^{F, \text{prop}}  & =  ( \varepsilon_t X_C^\top \E{\varepsilon_t^2}^{-1} ) \E{ \E{\varepsilon_t^2}^{-1} X_C \varepsilon_t^2 X_C^\top \E{\varepsilon_t^2}^{-1} }^{-1} \\
& = \varepsilon_tX_C^\top \Var{X_C}^{-1} = \varepsilon_t \varepsilon_C^\top \Var{\varepsilon_C}^{-1}.
\end{align*}
In this step, we have orthogonalized the components within $ \varphi_{\beta_{tC}}^{F, \text{prop}}  $ such that we find 
\[
\E{ S_{\beta_{tC}}^\top  \varphi_{\beta_{tC}}^{F, \text{prop}}} = \E{ \E{S_{\beta_{tC}}^\top  \varepsilon_t \, \middle| \varepsilon_C }  \varepsilon_C^\top \Var{\varepsilon_C}^{-1} }  = \E{ X_C  \varepsilon_C^\top \Var{\varepsilon_C}^{-1} } = \mathrm{Id},
\] 
where $ \E{S_{\beta_{tC}}^\top  \varepsilon_t \, \middle| \varepsilon_C }  = X_C $ follows as a result of $  \E{S_{\beta_{tC}}^\top  \varepsilon_t }  = X_C $ from Equation \eqref{equation_expectation_score_error}. Using the relation $ \E{S_{\beta_{ij}}^\top \varepsilon_t }  = 0 $ for all $ i \neq t $ or $ j \not\in C $ we obtain $ \E{ S_{\beta_{ij}}^\top \varphi_{\beta_{tC}}^{F, \text{prop}}}  = 0 $ for all remaining parameter components. So $ \varphi_{\beta_{tC}}^{F, \text{prop}}$ is an influence function by \cite{Tsiatis_Semiparametric}, Thm. 4.2. As $ \varphi_{\beta_{tC}}^{F, \text{prop}}$ lies after projection in the tangent space $ \mathscr{T}$, it must be the efficient influence function by \citep{Tsiatis_Semiparametric}, Thm 4.3, proving the efficiency of $ \varphi_{\beta_{tC}}^{F, \text{prop}} = \varphi_{\beta_{tC}}^{F, \text{eff}} $ in Model $ \mathcal{M}_1 $ for the parameter $ \beta_{tC}$.  Analogously we obtain $ \varphi_{\beta_{Mt}}^{F, \text{eff}} = \varepsilon_M \varepsilon_t \Var{\varepsilon_t}^{-1} $. 
When we however focus on the remaining components of $ \beta $, say $ \beta_{MC} $, scores of other $ \beta $ components corresponding to incoming edges of $ M $, so $ S_{\beta_{Mt} } $ have non-zero expectation with $ \varepsilon_M $. Correcting for these components is the orthogonalization in the parameter space. In particular, the efficient influence function for $ \beta_{MC} $ reads as
\[
\varphi_{\beta_{MC}}^{F, \text{eff}} = \varepsilon_M ( \varepsilon_C^\top \Var{\varepsilon_C}^{-1} - \varepsilon_t \Var{\varepsilon_t}^{-1} \beta_{tC}). 
\]
For any $ i \in [d_M], j \in [d_C] $, we obtain, with the unit vectors $ e_i, e_j $, 
\begin{align*}
\E{ S_{\beta_{M_i C_j}} \varphi_{\beta_{MC}}^{F, \text{eff}} } & = \E{ \E{ S_{\beta_{M_iC_j}} \varepsilon_M \, \middle| \, \varepsilon_{C,t}} ( \varepsilon_C^\top \Var{\varepsilon_C}^{-1} - \varepsilon_t \Var{\varepsilon_t}^{-1} \beta_{tC} ) }  \\ 
& = \E{ e_i e_j^\top X_C  ( \varepsilon_C^\top \Var{\varepsilon_C}^{-1} - \varepsilon_t \Var{\varepsilon_t}^{-1} \beta_{tC} ) }  = e_i e_j^\top
\end{align*}
as $ \E{X_C \varepsilon_t} = 0 $. Now we can inspect the inner product of $ \varphi_{\beta_{MC}}^{F, \text{eff}} $  with the score $ S_{\beta_{Mt}} $ to find
\begin{align*}
\E{ S_{\beta_{M_it}} \varphi_{\beta_{MC}}^{F, \text{eff}} } & = \E{ \E{ S_{\beta_{M_i t}} \varepsilon_M \, \middle| \, \varepsilon_{C,t}} ( \varepsilon_C^\top \Var{\varepsilon_C}^{-1} - \varepsilon_t \Var{\varepsilon_t}^{-1} \beta_{tC} ) }  \\ 
& = \E{ e_i X_t  ( \varepsilon_C^\top \Var{\varepsilon_C}^{-1} - \varepsilon_t \Var{\varepsilon_t}^{-1} \beta_{tC} ) } \\ 
& = \E{ e_i ( \beta_{tC} \varepsilon_C + \varepsilon_t)  ( \varepsilon_C^\top \Var{\varepsilon_C}^{-1} - \varepsilon_t \Var{\varepsilon_t}^{-1} \beta_{tC} ) }  = 0. 
\end{align*}
The crucial point is the fact that $ X_t $ decomposes to $\beta_{tC} \varepsilon_C + \varepsilon_t $. Using that the expectation of any other score component multiplied with $ \varepsilon_M $ is zero, we have thus shown that the proposed function is orthogonal in the parameter nuisance space. As $ \varphi_{\beta_{MC}}^{F, \text{eff}} $ lies as a result of linearity in the nuisance tangent space $ \mathscr{T} $, it must be efficient for $ \beta_{MC} $ in $ \mathcal{M}_1 $. Using $ X_j = \mathrm{E}_{\varepsilon_r | \varepsilon_t}[ \varepsilon_r^\perp S_{\beta_{rj}} | \varepsilon_t ] $ we find $ \varphi_{\beta_{rM}}^{F, \text{eff}} = \varepsilon_r^\perp \varepsilon_M^\top \Var{\varepsilon_M}^{-1} $, which is also orthogonal to all functions in $ \Lambda_\varepsilon^\perp $. 
Analogously proceeding, we summarize the full-data efficient influence function for $ \beta_{(t,M,r)(C,t,M)} $ in $ \mathcal{M}_1 $ to be $ \varphi_{\beta_{(t,M,r)(C,t,M)}}^{F, \text{eff}} = $ 
\[ \resizebox{\textwidth}{!}{ $
\begin{pmatrix}
\varepsilon_t \varepsilon_C^\top \Var{\varepsilon_C}^{-1} & 0 & 0  \\ 
\varepsilon_M ( \varepsilon_C^\top \Var{\varepsilon_C}^{-1} - \varepsilon_t \Var{\varepsilon_t}^{-1} \beta_{tC}) &  \varepsilon_M \varepsilon_t \Var{\varepsilon_t}^{-1} & 0 \\ 
\varepsilon_r^\perp (\varepsilon_C^\top \Var{\varepsilon_C}^{-1} - \varepsilon_M^\top \Var{\varepsilon_M}^{-1}(\beta_{Mt}\beta_{tC} + \beta_{MC})) & 0 &  \varepsilon_r^\perp \varepsilon_M^\top \Var{\varepsilon_M}^{-1}  \\ 
\end{pmatrix}. $}\]
The component $ \beta_{S,} $ is irrelevant in the efficient estimation of $ \xi $ and can therefore be considered as nuisance; see also \cite{bhattacharya2020identification}. 
Finally, by linearity of the causal effect $ \xi = \beta_{rM} \beta_{Mt} $, the efficient influence function for $ \xi $ must be \[ \varphi_{\xi}^{F, \text{eff}} = \beta_{rM} \varepsilon_M \varepsilon_t \Var{\varepsilon_t}^{-1} + \varepsilon_r^\perp \varepsilon_M^\top \Var{\varepsilon_M}^{-1} \beta_{Mt} ,\]
the corresponding linear combination of the efficient influence functions for $ \beta_{rM} $ and $ \beta_{Mt} $. This function is orthogonal to all nuisance directions and therefore efficient in $ \mathcal{M}_1 $.
\end{proof}

\subsection{Proof of Lemma \ref{lemma_observed_data_EIF}: Efficient Influence Function in Model \texorpdfstring{$\mathcal{M}_\pi$}{M\_pi}}
\restatement{Lemma \ref{lemma_observed_data_EIF}}{\TextLemmaObservedDataEIF}
\begin{proof}
We define the coarsened model tangent space $ \Lambda_\pi $ as the tangent space in Model $ \mathcal{M}_\pi $ with respect to the propensity function $ \pi $. The corresponding observed-data augmentation space is 
\begin{align*}
\Lambda_{2,\pi} = &\left\{\sum_{\delta \in \{ 1, 2\}}  \frac{I(\Delta = \delta) - (1 - \pi_\delta)I(\Delta \geq \delta)}{\Pi_{j = 1}^\delta \pi_j} (h_\delta \circ G_\delta)  \, \middle| \, \right. \\ 
&\qquad \qquad \left. h_1 \in L_0^2(\mathbb{R}^{d_C + 1}, \mathbb{R}), h_2 \in L_0^2(\mathbb{R}^{d_C + 1 + d_M}, \mathbb{R}) \right\}. 
\end{align*}
For any full-data influence function $ \varphi^F $, the set of all corresponding augmented observed-data influence functions in Model $ \mathcal{M}_\pi $ is the functional space
\[
IF^\pi(\varphi^F) =  \left\{ \left[\frac{I(\Delta = \infty)\varphi^F}{\pi_1 \pi_2} + h \right] - \Pi\left([\, \cdot \,] |  \Lambda_\pi \right) \, \middle| \, \forall h \in \Lambda_{2,\pi} \right\}.
\]
Applying \cite{Tsiatis_Semiparametric}, Thm 10.1 and Thm 10.4, to our setting yields the variance-minimizing, optimal element in $ IF^\pi(\varphi^F) $. Denoting this optimal augmented full-data influence function by $ \varphi^\ast $, its observed-data representation is
\begin{align*}
\mathcal{J}(\varphi^\ast) &= \frac{I(\Delta = \infty)\varphi^\ast}{\pi_1 \pi_2 } + \frac{I(\Delta = 1) -  (1 - \pi_1) I(\Delta \geq 1)}{\pi_1} \E{ \varphi^\ast \mid \varepsilon_{C,t}} \\ 
& \quad + \frac{I(\Delta = 2) - (1 - \pi_2) I(\Delta \geq 2)}{\pi_1 \pi_2} \E{ \varphi^\ast \mid \varepsilon_{C,t,M}}.
\end{align*}
Each component in the full-data efficient influence function $ \varphi_{\beta}^{F,\text{eff}} $ satisfies \[
\E{ \varphi_{\beta_{ij}}^{F,\text{eff}} | \varepsilon_{C,t}} \in \{ 0, \varphi_{\beta_{ij}}^{F,\text{eff}}\}, \quad \E{ \varphi_{\beta_{ij}}^{F,\text{eff}} | \varepsilon_{C,t,M}} \in \{ 0, \varphi_{\beta_{ij}}^{F,\text{eff}} \}  \quad,  i,j \in [d] .
\]
Consequently, the optimal full-data influence function within the augmentation space $ IF^\pi(\varphi_{\beta}^{F,\text{eff}}) $, and analogously within $ IF^\pi(\varphi_{\xi}^{F,\text{eff}}) $, simplifies~to
\[
\varphi_{\beta_{(t,M,r)(C,t,M)}}^\text{{opt}}(X; \pi)  = \begin{pmatrix}
\text{diag}(1) & 0 & 0  \\
0 & \frac{\text{diag}\left(I(\Delta \geq 2)\right)}{\pi_1(X_{C,t})}& 0 \\
0 & 0 & \frac{\text{diag}\left(I(\Delta = \infty)\right)}{\pi_1(X_{C,t})\pi_2(X_{C,t,M})} \\
\end{pmatrix} \varphi^{F,\text{eff}}_{\beta_{(t,M,r)(C,t,M)}}(X) 
\]
and 
\[
\varphi_{\xi}^\text{{opt}}(X; \pi)  = \frac{I(\Delta \geq 2)\beta_{rM} \varepsilon_M \varepsilon_t \Var{\varepsilon_t}^{-1} }{\pi_1(X_{C,t})} + \frac{I(\Delta = \infty)\varepsilon_r^\perp \varepsilon_M^\top \Var{\varepsilon_M}^{-1} \beta_{Mt}}{\pi_1(X_{C,t})\pi_2(X_{C,t,M})}.
\]
It is important to note that $ \varphi_{\beta} $ might not be the efficient influence function in the observed-data Model $ \mathcal{M}_\pi $. Every full-data influence function lies in the affine space $ \varphi_{\beta}^{F,\text{eff}} + \Lambda_\varepsilon^\perp $ as characterized in Equation \eqref{formula_nuisance_orthogonal_compelement}. Because the operator $ \mathcal{J}$ is linear, the efficient observed-data influence function for $\beta $ must take the form $ \varphi_{\beta}^{F,\text{eff}} + u $, where $ u \in \Lambda_\varepsilon^\perp $ is chosen to minimize the variance 
\[ \Var{\mathcal{J}(\varphi^\text{F,eff}_\beta)} + \Var{\mathcal{J}(u)} + 2 \mathrm{Cov}\left(\mathcal{J}(u), \mathcal{J}(\varphi^\text{F,eff}_\beta)\right). \]
\end{proof}

\subsection{Proof of Lemma \ref{lemma_linear_estimator}: Optimized Estimators are Solutions to Linear Equations} \label{appendix_proof_nested_least_squares}
\restatement{Lemma \ref{lemma_linear_estimator}}{\TextLemmaEfficientEstimator}
\begin{proof}
We prove the statement by deriving the optimized estimator for $ \beta_{(t,M,r)(C,t,M)} $ and subsequently applying the delta method to $ \xi = \phi(\beta) \equiv \beta_{rM} \beta_{Mt} $.
Let $ \hat{\beta}_n $ denote the solution to the $ d^2 $ estimating equations $ 0^{d^2} = \sum_{i} \varphi^{\text{opt},\pi}_{\beta} (X^{(i)}; \theta)$ on the  observed-data influence $ \varphi^{\text{opt},\pi}_{\beta} $ function from Lemma \ref{lemma_observed_data_EIF}. This estimator satisfies $ \sqrt{n} ( \hat{\beta}_n -  \beta ) \overset{P_{\theta; \pi}}{\longrightarrow} \mathcal{N}_{d^2} \left( 0, \textrm{Var}(\varphi^{\text{opt},\pi}_{\beta}) \right) $ (\cite{Tsiatis_Semiparametric}, Chapter~3.3). Using the structural equation model $ X = \beta X + \varepsilon $,  the relevant estimating equations take the form
\[
\resizebox{\textwidth}{!}{$
\begin{pmatrix}
0^{d_C} \\ 0^{d_M} \\ 0^{d_M \times d_C} \\ 0^{1 \times d_M} \\ 0^{1 \times d_C} 
\end{pmatrix} = \begin{pmatrix}
\sum_{i = 1}^n  \left(X_t^{(i)} - \beta_{Ct} X_{C}^{(i)}  \right) X_C^{(i)\top} \Var{\varepsilon_C}^{-1} \\ 
\sum_{i = 1}^n  \frac{I(\Delta \geq 2)}{\pi_1(X_{C,t}^{(i)})} \left(X_{M}^{(i)} - \beta_{Mt} X_{t}^{(i)} - \beta_{MC} X_C^{(i)} \right) \hat\varepsilon_t^{(i)} \Var{\varepsilon_t}^{-1} \\ 
\sum_{i = 1}^n  \frac{I(\Delta \geq 2)}{\pi_1(X_{C,t,M}^{(i)})} \left(X_{M}^{(i)} - \beta_{Mt} X_{t}^{(i)} - \beta_{MC} X_C^{(i)} \right)^\top \left(X_{C}^{(i)\top} \Var{\varepsilon_C}^{-1}  - \hat\varepsilon_t^{(i)} \Var{\varepsilon_t}^{-1} \beta_{tC}  \right) \\
\sum_{i = 1}^n  \frac{I(\Delta = \infty)}{\pi_1(X_{C,t}^{(i)})\pi_2(X_{C,t,M}^{(i)})} \left(X_{r}^{(i)} - \beta_{rM} X_{M}^{(i)} - \beta_{rC} X_C^{(i)} - \E{X_{r}^{(i)} - \beta_{rM} X_{M}^{(i)} - \beta_{rC} X_C^{(i)} | \varepsilon_t} \right)^\top \hat\varepsilon_M^{(i)} \Var{\varepsilon_M}^{-1} \\
\sum_{i = 1}^n  \frac{I(\Delta = \infty)}{\pi_1(X_{C,t}^{(i)})\pi_2(X_{C,t,M}^{(i)})} \left(X_{r}^{(i)} - \beta_{rM} X_{M}^{(i)} - \beta_{rC} X_C^{(i)} - \E{X_{r}^{(i)} - \beta_{rM} X_{M}^{(i)} - \beta_{rC} X_C^{(i)} | \varepsilon_t} \right)^\top \left(X_{C}^{(i)\top} \Var{\varepsilon_C}^{-1}  - \hat\varepsilon_M^{(i)} \Var{\varepsilon_M}^{-1} (\beta_{Mt} \beta_{tC} + \beta_{MC})  \right)  
\end{pmatrix}.
$}
\]
Each block corresponds to a linear equation system and can be solved iteratively. The first block uniquely determines the closed-form estimator for $ \beta_{tC} $ to be \[
\hat\beta_{tC, n} =  \left(\sum_{i=1}^{n} X_t^{(i)} X_C^{(i)\top}\right) \left(\sum_{i=1}^{n} X_C^{(i)} X_C^{(i)\top}\right)^{-1}.
\] 
This allows us to replace to specify all treatment error terms as $ \hat{\varepsilon}_t^{(i)} = X_t^{(i)} - \hat\beta_{Ct, n} X_{C}^{(i)} $. The next two equation systems have to be solved simultaneously. Let $ W_i = {I(\Delta \geq 2)}/{\pi_1(X_{C,t}^{(i)})}$. Define the weighted empirical correlations
\begin{align*}
 &E_{Mt} = \sum_{i = 1}^n W_i X_{M}^{(i)} \hat\varepsilon_t^{(i)}, 
 &&E_{tt} = \sum_{i = 1}^n W_i X_{t}^{(i)} \hat\varepsilon_t^{(i)}, 
 &&E_{Ct} = \sum_{i = 1}^n W_i X_{C}^{(i)} \hat\varepsilon_t^{(i)}, \\ 
 &E_{MC} = \sum_{i = 1}^n W_i X_{M}^{(i)} \hat\varepsilon_C^{(i)\top}, 
 &&E_{tC} = \sum_{i = 1}^n W_i X_{t}^{(i)} \hat\varepsilon_C^{(i)\top}, 
 &&E_{CC} = \sum_{i = 1}^n W_i X_{C}^{(i)} \hat\varepsilon_C^{(i)\top}.
\end{align*}
Subtracting $ \beta_{tC} $ times the second system from the third system yields the joint solution to the linear equation
\[
\begin{pmatrix}
\hat\beta_{Mt,n} \\ 
\hat\beta_{MC,n} 
\end{pmatrix} = 
\begin{pmatrix}
E_{tt} & E_{Ct} \\ 
E_{tC}  & E_{CC} 
\end{pmatrix}^{-1}
\begin{pmatrix}
E_{Mt} \\ 
E_{MC} 
\end{pmatrix}.
\]
The mediating residuals are $ \hat\varepsilon_M^{(i)} = X_{M}^{(i)} - \hat{\beta}_{Mt, n} X_{t}^{(i)} - \hat{\beta}_{MC, n} X_C^{(i)} $. For the final set of equations, we incorporate the orthogonality restriction between $ \varepsilon_t $ and $ \varepsilon_r^\perp $. Since $\varphi_\beta^\text{opt} $ is orthogonal to all elements in $ \Lambda_{r|t} $, only $\varepsilon_t$-linear components of $ \E{\varepsilon_r^{(i)} | \varepsilon_t} $ contribute to the estimation. Hence, we can replace this conditional expectation by its best linear predictor $ \gamma \hat\varepsilon_t^{(i)} $. Let $ W_i = {I(\Delta = \infty)}/{\pi_1(X_{C,t}^{(i)})\pi_2(X_{C,t,M}^{(i)})}$ and define $ Z^{(i)}_1 = ( X_C^{(i)}, \hat\varepsilon_t^{(i)}, \hat\varepsilon_M^{(i)}) $, $ Z^{(i)}_2 = ( X_C^{(i)}, \hat\varepsilon_t^{(i)}, X_M^{(i)}) $, and $ \theta = (\beta_{rC}^\top, \gamma, \beta_{rM}^\top)^\top $. The last two equation systems for the estimators of $ \beta_{rC} $ and $ \beta_{rM}$ are jointly solved by
\[
\begin{pmatrix}
\sum_i W_i X_C^{(i)} (X_r^{(i)} - Z^{(i)\top}_2 \theta )  = 0^{d_C} \\ 
\sum_i W_i \hat\varepsilon_t^{(i)} (X_r^{(i)} - Z^{(i)\top}_2 \theta )  = 0 \\ 
\sum_i W_i \hat\varepsilon_M^{(i)} (X_r^{(i)} - Z^{(i)\top}_2 \theta )  = 0^{d_M}
\end{pmatrix}
\Leftrightarrow
\begin{pmatrix}
\hat\beta_{rC,n} \\
\hat\gamma_n \\  
\hat\beta_{rM,n}
\end{pmatrix} = (\sum_i W_i Z^{(i)}_1 Z^{(i)\top}_2)^{-1}  (\sum_i W_i Z^{(i)}_1 X_r^{(i)}).
\]
Using orthogonality of the efficient influence function components, i.e., $ \mathrm{Cov}\left( \varphi^{\text{opt},\pi}_{\beta_{Mt}}, \varphi^{\text{opt},\pi\top}_{\beta_{rM}}  \right) = 0 $, we obtain the joint asymptotic distribution
\[ 
\sqrt{n}\left(\begin{pmatrix}
\hat\beta_{Mt, n} \\ \hat\beta_{rM, n}^\top 
\end{pmatrix} - \begin{pmatrix}
\beta_{Mt} \\ \beta_{rM}^\top 
\end{pmatrix}\right) \overset{P_{\theta; \pi}}{\longrightarrow}
\mathcal{N}_{2d_M} \left(
0, 
\begin{pmatrix}
\Var{\varphi^{\text{opt},\pi}_{\beta_{Mt}}} & 0 \\ 0 & \Var{\varphi^{\text{opt},\pi}_{\beta_{rM}}} 
\end{pmatrix} \right). \]
Applying the delta method to $ \hat\xi_n =  \hat\beta_{rM, n} \hat\beta_{Mt,n} $ with $ \xi_n =  \beta_{rM} \hat\beta_{Mt} $ yields the asymptotic 
\[
 \sqrt{n}(\hat{\xi}_n - \xi) \overset{P_{\theta; \pi}}{\longrightarrow} \mathcal{N}\left(0, \nabla \phi(\beta)^{\top}  \begin{pmatrix}
\Var{\varphi^{\text{opt},\pi}_{\beta_{Mt}}} & 0 \\ 0 & \Var{\varphi^{\text{opt},\pi}_{\beta_{rM}}} 
\end{pmatrix}  \nabla \phi(\beta) \right),
\]
where $ \nabla \phi(\beta)^{\top} = (\beta_{rM}, \beta_{Mt}^\top)  $. Thus the asymptotic variance of $ \hat\xi_n = \hat\beta_{rM, n} \hat\beta_{Mt,n}$ coincides with the derived efficiency bound  $ \Var{\varphi_\xi^{\text{opt},\pi}} $ within $ IF^\pi(\varphi_\xi^{F,\text{eff}}) $, as stated in Lemma \ref{lemma_observed_data_EIF}.
\end{proof}

\subsection{Proof of Lemma \ref{lemma_effect_EIF_var}: Variance of Efficient Effect Estimator in Model \texorpdfstring{$\mathcal{M}_\pi$}{M\_pi}}
\restatement{Lemma \ref{lemma_effect_EIF_var}}{\TextLemmaEffectEIFVar}
\begin{proof}
Using the coarsening probabilities  $ \mathcal{P}\left( \Delta \geq 2  \,  | \, \varepsilon_{C,t} \right) = \pi_1(X_{C,t}) $,  $ \mathcal{P}\left( \Delta = \infty  \,  | \, \varepsilon_{C,t,M} \right) = \pi_1(X_{C,t})\pi_2(X_{C,t,M}) $, and the observation that $ \E{\varepsilon_r^\perp \varepsilon_t \varepsilon_M  \varepsilon_M^\top }  = 0 $ is cancelling mixed terms, we get 
\begin{align*}
& \Var{ \varphi_\xi^{\text{opt},\pi}} = \\ 
& \quad = \E{ \frac{I(\Delta \geq 2)}{\pi_1^2} \left(\beta_{rM} \varepsilon_M \varepsilon_t \Var{\varepsilon_t}^{-1}\right)^2 } + \E{ \frac{I(\Delta = \infty)}{\pi_1^2\pi_2^2} \left(\varepsilon_r^\perp \varepsilon_M^\top \Var{\varepsilon_M}^{-1} \beta_{Mt}\right)^2 } \\  
& \quad = \E{ \frac{\varepsilon_t^2}{ \Var{\varepsilon_t}^{2} \pi_1^2} \E{ I(\Delta \geq 2)  \,  | \, \varepsilon_{C,t} } \beta_{rM} \E{ \varepsilon_M  \varepsilon_M^\top  \,  | \, \varepsilon_{C,t} } \beta_{rM}^\top }\\ 
& \qquad +  \E{ \frac{\left(\varepsilon_M^\top \Var{\varepsilon_M}^{-1} \beta_{Mt}\right)^2}{\pi_1^2\pi_2^2} \E{ I(\Delta = \infty)  \,  | \, \varepsilon_{C,t,M} } \E{ \varepsilon_r^{\perp 2}\,  | \, \varepsilon_{C,t,M} }   } \\  
& \quad = \E{ \frac{\varepsilon_t^2 \beta_{rM}  \Var{\varepsilon_M} \beta_{rM}^\top }{ \Var{\varepsilon_t}^{2} \pi_1(X_{C,t})} } +  \E{ \frac{\left(\varepsilon_M^\top \Var{\varepsilon_M}^{-1} \beta_{Mt}\right)^2 \Var{ \varepsilon_r\,  | \, \varepsilon_{t} } }{\pi_1(X_{C,t})\pi_2(X_{C,t,M})}   } . 
\end{align*}
\end{proof}

\subsection{Proof of Theorem \ref{theorem_optimization}: Proof of the Optimization Problem}
\restatement{Theorem \ref{theorem_optimization}}{\TextTheoremOptimization}
\begin{proof}
The functions $ g_1$ and $ g_2 $ are defined so that $ \Var{\varphi_\xi^{\text{opt}}} = \mathrm{E}[g_1 / \pi_1 + g_2 / (\pi_1 \pi_2)]$.
Let $ \lambda \in \mathbb{R} $ and let $ \nu_1, \nu_2, \kappa_1, \kappa_2 $ be measurable functions. 
The Lagrangian $ \mathcal{L}(\pi, \lambda, \nu , \kappa)  $ associated with this optimization problem is
\begin{align} \label{equation_lagrangian}
 \mathcal{L}(\pi, \lambda, \nu , \kappa)  := \int & \frac{g_1 }{\pi_1} + \frac{g_2}{\pi_1 \pi_2} + \lambda (c_1 \pi_1 + c_2 \pi_1 \pi_2 - b_0 ) - \nu_1 \pi_1+ \kappa_1 (\pi_1 - 1) \\ & \qquad - \nu_2 \pi_2+ \kappa_2 (\pi_2 - 1) \; \text{d} \mathcal{P}_{\theta} ,
\end{align}
where we suppress the index $ \mathcal{P}_{\theta, \pi}  $ and write $ \mathcal{P}_{\theta}$, since $ \pi $ affects only the distribution of the sampling stage variable $ \Delta $ \citep{Ito2008}.  First, we consider a variation in the second-stage propensity. For $ h \in \mathcal L ^2 $, define $ \pi_{\gamma; 2} = (\pi_1, \pi_2 + \gamma h) $ and compute the derivative of the Lagrangian at $ \gamma = 0 $:
$$
\left. \frac{ \partial \mathcal L (\pi_{\gamma;2}, \lambda, \nu , \kappa)}{ \partial \gamma}\right\vert_{\gamma = 0} =  \int \left(  -\frac{g_2}{\pi_1\pi_2^2} + \lambda c_2 \pi_1 - \nu_2 + \kappa_2 \right) h \; \text{d} \mathcal{P}_{\theta}
$$
For any optimal $ \pi $, this expression must vanish for all $ h $, and therefore the bracketed term must vanish almost everywhere. For interior points of the constraint set, complementary slackness yields the pointwise condition
\begin{align} \label{equation_variation2}
-\frac{g_2}{\pi_1\pi_2^2} + \lambda c_2 \pi_1 = 0 \quad \Leftrightarrow \quad \pi_2^2 = \frac{g_2}{\lambda c_2 \pi_1^2} . 
\end{align}
Analogously, a variation on $ \pi_1 $ yields the necessary condition 
\begin{align} \label{equation_variation1}
-\frac{g_1}{\pi_1^2} - \frac{g_2}{\pi_1^2 \pi_2} + \lambda (c_1 + c_2 \pi_2) = 0 \quad \Leftrightarrow \quad \pi_1^2 = \frac{g_1 + g_2 / \pi_2}{\lambda (c_1 + c_2 \pi_2)}.
\end{align}
Substituting the expression for $ \pi_2 $ from Equation \eqref{equation_variation2} gives \[ g_2 / \pi_2 = \sqrt{g_2 \lambda c_2} \pi_1,  c_2 \pi_2 = \sqrt{g_2 c_2} / (\pi_1 \sqrt{\lambda}) \] and therefore Equation \eqref{equation_variation1} simplifies to $ \pi_1^2 = g_1 / (\lambda c_1) =: \tilde{\pi}_1^2 > 0 $. Using this representation in Equation \eqref{equation_variation2} yields $ \pi_2^2 = (g_2 c_1) / (g_1 c_2) =: \tilde{\pi}_2^2 > 0 $. Whenever both $ \tilde{\pi}_1^2 $ and $ \tilde{\pi}_2^2 $ lie in the interior of the feasible set, the pair $ \tilde{\pi}_{1,2} $ is optimal. If either exceeds 1, the corresponding component must be clipped at the boundary and the KKT conditions adjust accordingly. When $ \pi_2^\ast = 1 $ or $ \pi_1^\ast = 1 $ the Equations \eqref{equation_variation1} and \eqref{equation_variation2} reduce to $ \pi_1^2 = \frac{g_1 + g_2}{\lambda (c_1 + c_2)} $ and $ \pi_2^2 = \frac{g_2}{\lambda c_2} $, yielding mixed solutions. Finally, if both $ \tilde{\pi}_1^2 $ and $ \tilde{\pi}_1^2 $ exceed $ 1 $, no mixed solution is feasible, since the right-hand sides of Equations \eqref{equation_variation2} and \eqref{equation_variation1} are non-decreasing with decreasing as $ \pi_1 $ and $ \pi_2 $, decrease. Thus the solution is $ \pi^\ast = (1, 1) $. Collecting all cases, every solution of the Lagrangian \eqref{equation_lagrangian} must satisfy the following pointwise conditions:
\[
(\pi_1, \pi_2) = 
	 \begin{cases}
	  \left(\sqrt\frac{g_1}{\lambda c_1}, \sqrt\frac{g_2 c_1}{g_1 c_2} \right) & \text{, where } g_1 < \lambda c_1 \text{ , and } g_2 c_1 < g_1 c_2 \\ 
	  \left(\min\left(1, \sqrt{\frac{g_1 + g_2}{\lambda(c_1 + c_2)}}\right), 1\right) & \text{, where } g_1 < \lambda c_1 \text{ , and }  g_2 c_1 \geq g_1 c_2 \\ 
	  \left(1, \min\left(1, \sqrt{\frac{g_2}{\lambda c_2}}\right)\right) & \text{, where } g_1 \geq \lambda c_1  \\ 
 \end{cases}
\]
Rearranging terms and using continuity at $ g_2 c_1 = g_1 c_2$, this can be equivalently expressed as
\[ 
(\pi_1, \pi_2) = 
	 \begin{cases}
	  \left(\min\left(1, \max\left( \sqrt\frac{g_1}{\lambda c_1}, \sqrt{\frac{g_1 + g_2}{\lambda(c_1 + c_2)}}\right)\right), \min\left(1, \sqrt\frac{g_2 c_1}{g_1 c_2}\right) \right) & \text{, where } g_1 < \lambda c_1 \\ 
	  \left(1, \min\left(1, \sqrt{\frac{g_2}{\lambda c_2}}\right)\right) & \text{, where } g_1 \geq \lambda c_1 \\ 
 \end{cases}.
\] 
However, since $ c_2 $ and $ g_2 $ depend on $ \varepsilon_M $, they are not available when computing $ \pi_1 $. Thus the above representation it is not admissible. Its conditional version, 
\[
(\pi_1^\ast, \pi_2^\ast) = 
	 \begin{cases}
	  \left(\min\left(1, \max\left( \sqrt\frac{g_1}{\lambda c_1}, \sqrt{\frac{g_1 + \E{g_2| \varepsilon_{C,t}}}{\lambda(c_1 + \E{c_2| \varepsilon_{C,t}})}}\right)\right), \min\left(1, \sqrt\frac{g_2 c_1}{g_1 c_2}\right) \right) & \text{, } g_1 < \lambda c_1 \\ 
	  \left(1, \min\left(1, \sqrt{\frac{g_2}{\lambda c_2}}\right)\right) & \text{, } g_1 \geq \lambda c_1 \\ 
 \end{cases}
\]
is admissible and optimal. Note that $ (\pi_1^\ast, \pi_2^\ast) $ is continuous and $ \pi_2^\ast $ remains unchanged. Optimality follows by applying the law of iterated expectation to the Lagrangian in Equation \eqref{equation_lagrangian}. Finally, the optimal multiplier $ \lambda^\ast $ is determined by enforcing the budget constraint $ \mathrm{E}[c_0 + (\pi_1 c_1)(X_{C,t}) + \pi_1(X_{C,t})(\pi_2 c_2)(X_{C,t,M})] = b_0 $ which concludes the solution to a restricted version of Problem \ref{problem_main}.
\end{proof}

\section{Simulation Study. Setup and Experiments} \label{appendix_simulation_study}
This section describes the front-door model and the experiments analysed in Section \ref{section_simulation_study}. Our ground truth parameter matrix, with dimension specifications $ d_C = 2$, $ d_M = 3$, is given by 
\[
\beta = 
\left(
\begin{array}{cc|c|ccc|c|c}
0 & 0 & 0 & 0 & 0 & 0 & 0 & \cdot \\ 
0 & 0 & 0 & 0 & 0 & 0 & 0 & \cdot \\  \hline
0.5 & -0.2 & 0 & 0 & 0 & 0 & 0 & \cdot \\ \hline
0.3 & 0.1 & 0.7 & 0 & 0 & 0 & 0 & \cdot \\ 
0.5 & 0.2 & 0.2 & 0 & 0 & 0 & 0 & \cdot \\ 
-0.1 & 0.3 & 0.1 & 0 & 0 & 0 & 0 & \cdot \\ \hline
0.2 & -0.1 & 0 & 0.5 & 0.4 & -0.3 & 0 & \cdot \\ \hline
\cdot & \cdot & \cdot & \cdot & \cdot & \cdot & \cdot & \cdot
\end{array}
\right)
\]
This specifications determines the causal effect $ \xi $ to be $ \beta_{rM} \beta_{Mt} = 0.4$. The data is now generated by the structural equations $ X = (I - \beta)^{-1} \varepsilon $ where the error distributions have the form 
\begin{align*}
&
\varepsilon_C \sim t_2 \left( \begin{pmatrix} 0 \\ 0 \end{pmatrix}, \sigma_C, df = 5 \right) , \varepsilon_{tr} \sim t_2 \left( \begin{pmatrix} 0 \\ 0 \end{pmatrix}, \sigma_{tr}, df = 5 \right) \text{  and } 
\varepsilon_M \sim \mathcal{N}_3 \left( \begin{pmatrix} 0 \\ 0 \end{pmatrix}, \sigma_M \right) \text{ with} \\ 
&
\sigma_C = \begin{pmatrix} 1 & 0.7 \\  0.7 & 1.5 \end{pmatrix}, 
\sigma_{tr} = \begin{pmatrix} 1 & -0.5 \\ -0.5 & 1.5 \end{pmatrix} \text{, and }
\sigma_M = \begin{pmatrix} 1 & 0.3 & 0 \\ 0.3 & 1.5 & -0.5 \\ 0 & -0.5 & 1 \end{pmatrix}. 
\end{align*}
We fixed the cost functions to scale with the norm and to be constant, i.e., \[ c_1(x_{Ct}) = 0.1 \cdot \Vert x_{Ct} \Vert_2,  \quad c_2(x_{CtM}) = 0.5. \]
\textbf{Experiment to Figure \ref{figure_calibration}:} We created $ 50 $ replications of datasets with sample sizes \[N = (100, 250, 500, 750, 1000, 2500, 5000, 7500)\] under $ \pi_1 \equiv 1,  \pi_2 \equiv 1 $. For each dataset, the estimate $ \hat\xi_{\text{naive}} $, average sample cost $ c_{\text{naive}} $ and asymptotic variance $ v_{\text{naive}} $ under $ \pi_1 \equiv 1, \pi_2 \equiv 1  $ were determined. We computed the optimal propensities $ \pi^\ast $ at a budget of $ b_0 = c_{\text{naive}} / 1.5 $ and determined the estimate $ \hat\xi_{\text{opt}}$ and optimized asymptotic variance $ v_{\text{opt}}$ on dataset created by the optimal propensities $ \pi^\ast $ of sizes $ 1.5 \cdot N $ to counteract the smaller budget. In Figure \ref{figure_calibration}, the left plot portrays the empirical mean-squared errors $ MSE({\hat\xi_{\text{naive}}}) , MSE({\hat\xi_{\text{opt}}}) $ and their theoretical counterparts $ MSE_{\text{naive}} = v_{\text{naive}}, MSE_{\text{opt}} =v_{\text{opt}} $ at different budget levels $ n~\cdot~b_0 $. The right plot visualizes the distribution of the optimized propensity $ \pi^\ast $ on $ 1,000 $ samples. \\
\textbf{Experiment to Figure \ref{figure_computational_sensitivity}:} As in the experiment to Figure \ref{figure_calibration}, we generated 50 replications of datasets with sample sizes $ N $ under  $ \pi_1 \equiv 1,  \pi_2 \equiv 1 $. For each dataset, we optimized for $ {\lambda}^\ast $, computed the asymptotic variance $ v_{\text{opt}}$, its cost $ c_{\text{naive}} $ and stored the computational time $ \tau $. The four plots in Figure \ref{figure_computational_sensitivity} show the means at the various sample sizes along with the uncertainty measure of $ \pm $ one standard deviation. Note that all computations were run on an Apple M2 chip with 16GB of RAM. \\
\textbf{Experiment to Figure \ref{figure_sensitivity}:} In this experiment, we varied one parameter or nuisance of our data generation process at a time and studied its effects. For reference, we listed the exact alterations in Table \ref{table_alterations_experiment_parameter_nuisance_sensitivity}.
\begin{table}
\centering
\begin{tabular}{c|c}
Scaled Parameter / Nuisance & Multiplicative scaling range $ \alpha $ \\ \hline
$\beta_{tC} = \alpha \cdot (0.5, -0.2)$  & $ (\frac{1}{100}, \frac{5}{100}, \frac{10}{100}, \frac{15}{100}, \frac15, \frac14, \frac13, \frac12, 1, 2, 3, 4, 5, 10, 20) $   \\ 
$ \beta_{MC} = \alpha \cdot \begin{pmatrix} 0.3 & 0.1 \\ 0.5 & 0.2 \\ -0.1 & 0.3 \end{pmatrix}  $ & $ (\frac{1}{100}, \frac{5}{100}, \frac{10}{100}, \frac{15}{100}, \frac15, \frac14, \frac13, \frac12, 1, 2, 3, 4, 5, 10, 20) $ \\ 
$ \beta_{rC} = \alpha \cdot (0.2, -0.1) $ & $ (\frac{1}{100}, \frac{5}{100}, \frac{10}{100}, \frac{15}{100}, \frac15, \frac14, \frac13, \frac12, 1, 2, 3, 4, 5, 10, 20) $ \\ 
$ \beta_{Mt} = \alpha \cdot \begin{pmatrix} 0.7 \\ 0.2 \\ 0.1 \end{pmatrix} $ & $ (\frac{1}{100}, \frac{5}{100}, \frac{10}{100}, \frac{15}{100}, \frac15, \frac14, \frac13, \frac12, 1, 2, 3, 4, 5, 10, 20) $  \\
$ \beta_{rM} = \alpha \cdot (0.5, 0.4, -0.3)$ & $ (\frac{1}{100}, \frac{5}{100}, \frac{10}{100}, \frac{15}{100}, \frac15, \frac14, \frac13, \frac12, 1, 2, 3, 4, 5, 10, 20) $  \\ 
$ \sigma_C =  \begin{pmatrix}  1 \alpha & 0.7 \\ 0.7 &  1.5 \alpha\end{pmatrix}$  & $ (0.6, 0.75, 1, 1.3, 1.75, 2.25, 3, 4.5, 7, 10, 15) $  \\ 
$ \sigma_C =  \begin{pmatrix}  1  & 0.7\alpha \\ 0.7\alpha &  1.5 \end{pmatrix}$ & \verb| seq(0, sqrt(3), length.out = 11) |\\ 
$ \sigma_M = \begin{pmatrix} 1\alpha & 0.3 & 0 \\ 0.3 & 1.5\alpha & -0.5 \\ 0 & -0.5 & 1\alpha \end{pmatrix} $ & $ (0.5, 0.6, 0.8, 1.1, 1.6, 2.25, 3, 4.5, 7, 10, 15) $ \\ 
$ \sigma_M = \begin{pmatrix} 1 & 0.3\alpha & 0 \\ 0.3\alpha & 1.5 & -0.5\alpha \\ 0 & -0.5\alpha & 1 \end{pmatrix}$ &  \verb| seq(0, 2, by = 0.2) | \\ 
$ \sigma_{tr} = \begin{pmatrix} 1\alpha & -0.5 \\ -0.5 & 1.5 \end{pmatrix}     $ & $ (0.2, 0.4, 0.7, 1, 1.5, 2.5, 5, 8, 12.5, 20) $ \\ 
$ \sigma_{tr} = \begin{pmatrix} 1 & -0.5 \\ -0.5 & 1.5\alpha \end{pmatrix}$ & $ (0.3, 0.5, 0.7, 1, 1.5, 2.5, 5, 8, 12.5, 20)$ \\ 
$ \sigma_{tr} = \begin{pmatrix} 1 & -0.5\alpha \\ -0.5\alpha & 1.5 \end{pmatrix}$ & \verb| -seq(0, 1.2, length.out = 11) | \\ 
\end{tabular}
\caption{Alterations in the experiments to Figure \ref{figure_sensitivity}.}
\label{table_alterations_experiment_parameter_nuisance_sensitivity}
\end{table}
For each alteration, we generated $ 50 $ datasets of size $ 500 $ and computed the asymptotic variance $ v_{\text{naive}} $ and average sample cost $ c_{\text{naive}} $. The budget for the optimization problem was set to be $ b_0 = c_{\text{naive}} / 1.5 $ giving the asymptotic optimized variance $ v_{\text{opt}} $. Here, we did not generate a second larger dataset, but instead scaled the ratio of the asymptotic variances accordingly to get the relative efficiency $ v_{\text{opt}} / (1.5 \cdot v_{\text{naive}}) $. A relative efficiency of, say, $80\%$ indicates that under the same budget, the asymptotic optimized partial-measurements variance is $20\%$ smaller than the always-measuring asymptotic variance. The $ 12 $ plots showcase the means with the uncertainty of $ \pm $ one standard deviation against either the scaling factor $ \alpha $ or the resulting (average) correlation(s). Note that the scaling ranges were carefully chosen so as to generate positive definite matrices.
\\ 
\textbf{Computational Note:} To approximate the conditional expectations in Theorem \ref{theorem_optimization} at an evaluation point $ x_{C,t} $, we simulated realizations of $ X_{C,t,M} $ using a subsample of $ \varepsilon_M $, applied the functional, and then computed the weighted average of the subsample's propensity.

\section{Model Misspecification. Non-Linear Data} \label{appendix_misspecification}
This appendix extends the analysis underlying Figure \ref{figure_sensitivity} by examining efficiency gains under deliberate model misspecification. The data‑generating mechanism is identical to the linear setting (details in Appendix \ref{appendix_simulation_study}) except for the structural equation governing the mediator $ X_M$, which is replaced by the quadratic expression
\[ \begin{pmatrix}
    X_{M_1} \\ X_{M_2}  \\ X_{M_3} 
\end{pmatrix} = \begin{pmatrix}
    0.7 & -0.1 \\ 0.2 & -0.2 \\ 0.1 & -0.4
\end{pmatrix} 
\begin{pmatrix}
    X_t \\ X_t^2 
\end{pmatrix} + 
\begin{pmatrix}
 \varepsilon_{M_1} \\ \varepsilon_{M_2}  \\ \varepsilon_{M_3} 
\end{pmatrix}.
\]
Under this specification, the causal effect of $X_t$ on $X_r$ becomes a quadratic function of the treatment. To accommodate this, we estimate the causal effect using quadratic regression models while keeping the design optimization step unchanged. That is, the optimized partial‑data design is still constructed under the (incorrect) assumption of linearity. This experiment therefore analyzes the robustness of the optimized design when the estimation problem becomes non-linear.

We evaluate the causal effect at ten equally spaced treatment values in the interval $[-0.1, 0.1]$. Figure \ref{figure_reanalysis} reports the mean squared error (MSE), averaged over 50 replications and all evaluation points, across a range of budget levels. The display illustrates that the optimized partial‑data design continues to outperform the full‑data design, yielding smaller MSEs despite the slight misspecification. As the non‑linear component of the data‑generating process however becomes more pronounced, we expect the misalignment between the design stage and the estimation stage to increase, potentially reducing the efficiency gains.

\begin{figure}
    \centering
    \includegraphics[width=0.75\linewidth]{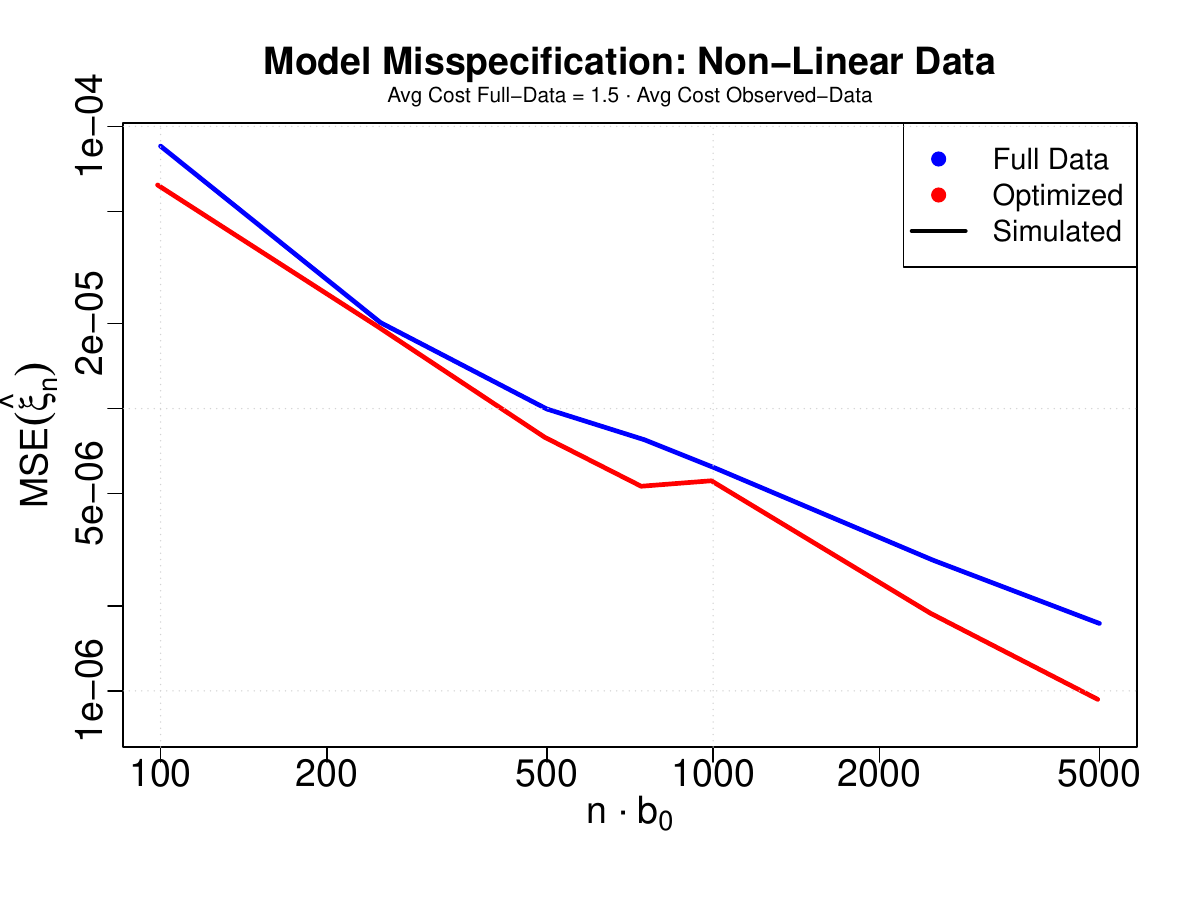}
    \caption{Mean squared error of the estimated causal effect under quadratic data‑generating mechanisms, evaluated and averaged at ten treatment values in $[-0.1, 0.1]$ over 50 replications.}
    \label{figure_reanalysis}
\end{figure}

\end{document}

%% file: rebuttal/Figure_generation.tex
\def\dashlength{3pt}
\begin{tikzpicture}[scale = 0.85, >=Stealth, node style/.style={scale = 0.8,  minimum height=1cm, font=\small}]
\node[node style, teal] (B) at (2,2) {$X_C $};
\node[node style] (t) at (0,0) {$X_t$};
\node[node style] (M) at (2,0) {$ X_M $};
\node[node style, olive] (c) at (0,-2) {$\Delta$};
\node[node style] (r) at (4,0) {$X_r$};

\draw[->, teal] (B) -- (t);
\draw[->, teal] (B) -- (M);
\draw[->, teal, dash pattern=on \dashlength off \dashlength, dash phase=0pt]  (B) to [bend right=50] (c);
\draw[->, olive, dash pattern=on \dashlength off \dashlength, dash phase=\dashlength] (B) to [bend right=50] (c);
\draw[->, teal] (B) -- (r);

\draw[<->] (t) to [out=-45,in=-135] (r);
\draw[->] (t) -- (M);
\draw[->, olive] (t) -- (c);

\draw[->, olive] (M) to [bend right=20]  (c);
\draw[->] (M) -- (r);
\filldraw[fill=white, fill opacity=0.1, draw=white, rounded corners=5pt]
  (-0.75,-3.625) rectangle (4.5, 3.625);
\end{tikzpicture} \hfill
\begin{tikzpicture}[scale = 0.85, >=Stealth, node style/.style={scale = 0.8, scale = 1, minimum height=1cm, align=left, anchor = north}]
\node[node style ] (B) at (4,2.5) {Measure $x_{Ct}$ at cost $c_0$};
\node[node style, blue, anchor = north west] (P0) at (-4,1.5) {With probability};
\node[node style, blue] (P1) at (1,1.5) {$\pi_1(x_{Ct})$};
\node[node style, blue] (P11) at (5.15,1.5) {$1 - \pi_1(x_{Ct})$};
\node[node style] (t) at (0,0) {$\begin{matrix} \textrm{Return } x_{Ct} \\ \Delta = 1 \end{matrix}$};
\node[node style] (M) at (4,0) {Measure $ x_M $ at cost $c_1(x_{Ct})$};
\node[node style, blue, anchor = north west] (P02) at (-4,-1.5) {With Probability };
\node[node style, blue] (P2) at (1,-1.5) {$\pi_2(x_{CtM})$};
\node[node style, blue] (P21) at (5.15, -1.5) {$1 - \pi_2(x_{CtM})$};
\node[node style, anchor=north] (r) at (0, -3) {$\begin{matrix} \textrm{Return } x_{CtM} \\ \Delta = 2 \end{matrix}$};
\node[node style, anchor = north] (S) at (4,-3) {$\begin{matrix} \textrm{Measure } x_{r} \textrm{ at cost } c_2(x_{CtM}) \\ \textrm{Return } x_{CtMr} \\ \Delta = \infty \end{matrix}$};

\def\dashlength{3pt}
\draw[->] (B) -- (t);
\draw[->] (B) -- (M);
\draw[->] (M) -- (r);
\draw[->] (M) -- (S);

\filldraw[fill=violet, fill opacity=0.1, draw=violet, rounded corners=5pt]
  (-4.5,-1.25) rectangle (7, 1.6);
\node[node style, violet, anchor = north west] (S1) at (-4, -0.5) {First Stage};
\filldraw[fill=orange, fill opacity=0.1, draw=orange, rounded corners=5pt]
  (-4.5,-4.5) rectangle (7, -1.5);
\node[node style, orange, anchor = north west] (S1) at (-4, -3.75) {Second Stage};
\filldraw[fill=white, fill opacity=0.1, draw=black, rounded corners=5pt]
  (-4.75,-4.75) rectangle (7.25, 2.5);
\end{tikzpicture}

%% file: Figure_Id_Graphs.tex
\begin{figure}
\center
\begin{minipage}{0.3\textwidth}
\begin{tikzpicture}[scale = 0.6, >=Stealth, node style/.style={scale = 0.6, minimum height=1cm, font=\small}]
\node[node style, teal] (B) at (2,2) {$X_C $};
\node[node style] (t) at (0,0) {$X_t$};
\node[node style] (M) at (2,0) {$ X_M $};
\node[node style, olive] (c) at (0,-2) {$\Delta$};
\node[node style] (r) at (4,0) {$X_r$};
\node[node style, violet] (S) at (4,-2) {$ X_S $};

\def\dashlength{3pt}
\draw[->, teal] (B) -- (t);
\draw[->, teal] (B) -- (M);
\draw[->, teal, dash pattern=on \dashlength off \dashlength, dash phase=0pt]  (B) to [bend right=50] (c);
\draw[->, olive, dash pattern=on \dashlength off \dashlength, dash phase=\dashlength] (B) to [bend right=50] (c);
\draw[->, teal] (B) -- (r);
\draw[<->, teal, dash pattern=on \dashlength off \dashlength, dash phase=0pt]  (B) to[bend left=50] (S);
\draw[<->, violet, dash pattern=on \dashlength off \dashlength, dash phase=\dashlength] (B) to[bend left=50] (S);

\draw[<->] (t) to [out=-45,in=-135] (r);
\draw[->] (t) -- (M);
\draw[->, olive] (t) -- (c);
\draw[->, violet] (t) to [bend right=30] (S);

\draw[->, olive] (M) to [bend right=20]  (c);
\draw[->] (M) -- (r);
\draw[->, violet] (M) to [bend left=20]  (S);
\draw[<->, dashed, violet] (M) to (S);

\draw[->, violet] (r) -- (S); 
\end{tikzpicture}
\end{minipage}
\begin{minipage}{0.3\textwidth}
\begin{tikzpicture}[scale = 0.6, >=Stealth, node style/.style={scale = 0.6, minimum height=1cm, font=\small}]
\node[node style, teal] (B) at (2,2) {$X_C$};
\node[node style] (t) at (0,0) {$X_t$};
\node[node style] (M) at (2,0) {$ X_M $};
\node[node style, olive] (c) at (0,-2) {$\Delta$};
\node[node style] (r) at (4,0) {$X_r$};
\node[node style, violet] (S) at (4,-2) {$ X_S $};

\def\dashlength{3pt}
\draw[->, teal] (B) -- (t);
\draw[<->, teal] (B) -- (M);
\draw[->, teal, dash pattern=on \dashlength off \dashlength, dash phase=0pt]  (B) to [bend right=50] (c);
\draw[->, olive, dash pattern=on \dashlength off \dashlength, dash phase=\dashlength] (B) to [bend right=50] (c);
\draw[->, teal] (B) -- (r);
\draw[->, teal, dash pattern=on \dashlength off \dashlength, dash phase=0pt]  (B) to[bend left=50] (S);
\draw[->, violet, dash pattern=on \dashlength off \dashlength, dash phase=\dashlength] (B) to[bend left=50] (S);

\draw[<->] (t) to [out=-45,in=-135] (r);
\draw[->] (t) -- (M);
\draw[->, olive] (t) -- (c);
\draw[->, violet] (t) to [bend right=30] (S);

\draw[->, olive] (M) to [bend right=20]  (c);
\draw[->] (M) -- (r);
\draw[->, violet] (M) to [bend left=20]  (S);

\draw[->, violet] (r) -- (S); 
\end{tikzpicture}
\end{minipage}
\begin{minipage}{0.3\textwidth}
\begin{tikzpicture}[scale = 0.6, >=Stealth, node style/.style={scale = 0.6, minimum height=1cm, font=\small}]
\node[node style, teal] (B) at (2,2) {$X_C$};
\node[node style] (t) at (0,0) {$X_t$};
\node[node style] (M) at (2,0) {$X_M $};
\node[node style, olive] (c) at (0,-2) {$\Delta$};
\node[node style] (r) at (4,0) {$X_r$};
\node[node style, violet] (S) at (4,-2) {$ X_S$};

\def\dashlength{3pt}
\draw[<->, teal] (B) to (t);
\draw[->, teal] (B) -- (M);
\draw[->, teal, dash pattern=on \dashlength off \dashlength, dash phase=0pt]  (B) to [bend right=50] (c);
\draw[->, olive, dash pattern=on \dashlength off \dashlength, dash phase=\dashlength] (B) to [bend right=50] (c);
\draw[<->, teal] (B) -- (r);
\draw[->, teal, dash pattern=on \dashlength off \dashlength, dash phase=0pt]  (B) to [bend left=50] (S);
\draw[->, violet, dash pattern=on \dashlength off \dashlength, dash phase=\dashlength] (B) to[bend left=50] (S);

\draw[<->] (t) to [out=-45,in=-135] (r);
\draw[->] (t) -- (M);
\draw[->, olive] (t) -- (c);
\draw[->, violet] (t) to [bend right=30] (S);

\draw[->, olive] (M) to [bend right=20]  (c);
\draw[->] (M) -- (r);
\draw[->, violet] (M) to [bend left=20]  (S);

\draw[->, violet] (r) -- (S); 
\end{tikzpicture}
\end{minipage}
\caption{These causal graphs allow for identification  \citep{drton_linear_identifiability} of linear causal effect $ \xi $. Directed edges result in non-zero elements in the design matrix $ B $, whereas bidirected edges indicate dependencies in the error vector $ \varepsilon $.}
\label{figure_identified_causal_graphs}
\end{figure}

%% file: rebuttal/Figure_proof_logic.tex
\begin{tikzpicture}[scale = 1, >=Stealth, node style/.style={scale = 1, minimum height=1cm, align=left, anchor = north}]
\node[node style ] (a) at (0, 2) {$\varphi^F_1$};
\node[node style ] (b) at (2, 2) {$\bm{\varphi^{F,\textbf{eff}}_2}$};
\node[node style ] (c) at (4, 2) {$\varphi^F_3$};
\node[node style ] (c) at (6, 2) {$\dots$};
\node[node style, violet, anchor = west, rotate=90] (S1) at (-5.5, 0.5) {$\begin{matrix} \textrm{Full-Data}  \\ \textrm{Influence} \\ \textrm{Functions} \end{matrix}$ };
\node[node style, anchor = west, rotate=90] (S2) at (-5.5, -2.4) {$\begin{matrix} \textrm{\color{red}Observed-Data}  \\ \textrm{\color{orange}Influence} \\ \textrm{\color{red}Functions} \end{matrix}$ };
\node[node style, anchor = west, rotate=90] (S2) at (-5.5, -5.5) {$\begin{matrix} \textrm{Feasible Set}  \\ \textrm{\color{black}Restriction} \\ \textrm{\color{black}in Thm. 11} \end{matrix}$ };

\node[node style ] (d) at (-4.25, 0) {$\varphi^{\pi^1}_{1.1}$};
\node[node style ] (e) at (-3.25, 0) {$\bm{\varphi^{\pi^1}_{1.2}}$};
\node[node style ] (f) at (-4.25, -0.5) {$\varphi^{\pi_1}_{1.3}$};
\node[node style ] (f0) at (-3.25, -0.5) {$\ddots$};
\node[node style ] (g) at (-2.25, 0) {$\varphi^{\pi^1}_{2.1}$};
\node[node style ] (h) at (-1.25, 0) {$\bm{\varphi^{\pi^1}_{2.2}}$};
\node[node style ] (i) at (-2.25, -0.5) {$\varphi^{\pi^1}_{2.3}$};
\node[node style ] (i0) at (-1.25, -0.5) {$\ddots$};
\node[node style ] (j) at (-0.25, 0) {$\varphi^{\pi^1}_{3.1}$};
\node[node style ] (k) at (0.75, 0) {$\bm{\varphi^{\pi^1}_{3.2}}$};
\node[node style ] (l) at (-0.25, -0.5) {$\varphi^{\pi^1}_{3.3}$};
\node[node style ] (l0) at (0.75, -0.5) {$\ddots$};
\node[node style ] (l1) at (1.65, -0.5) {$\dots$};

\node[node style ] (m) at (-4.25+7, 0) {$\varphi^{\pi^2}_{1.1}$};
\node[node style ] (n) at (-3.25+7, 0) {$\varphi^{\pi^2}_{1.2}$};
\node[node style ] (o) at (-4.25+7, -0.5) {$\bm{\varphi^{\pi^2}_{1.3}}$};
\node[node style ] (o0) at (-3.25+7, -0.5) {$\ddots$};
\node[node style ] (p) at (-2.25+7, 0) {$\varphi^{\pi_2}_{2.1}$};
\node[node style ] (q) at (-1.25+7, 0) {$\varphi^{\pi^2}_{2.2}$};
\node[node style ] (r) at (-2.25+7, -0.5) {$\bm{\varphi^{\pi_2}_{2.3}}$};
\node[node style ] (r) at (-1.25+7, -0.5) {$\ddots$};
\node[node style ] (s) at (-0.25+7, 0) {$\varphi^{\pi^2}_{3.1}$};
\node[node style ] (t) at (0.75+7, 0) {$\varphi^{\pi^2}_{3.2}$};
\node[node style ] (u) at (-0.25+7, -0.5) {$\bm{\varphi^{\pi^2}_{3.3}}$};
\node[node style ] (u0) at (0.75+7, -0.5) {$\ddots$};
\node[node style ] (u1) at (1.65+7, -0.5) {$\dots$};

\node[node style ] (g) at (-2.25, -3.5) {$\varphi^{\pi^1}_{2.1}$};
\node[node style ] (h) at (-1.25, -3.5) {$\bm{\varphi^{\pi^1}_{2.2}}$};
\node[node style ] (i) at (-2.25, -4) {$\varphi^{\pi^1}_{2.3}$};
\node[node style ] (i0) at (-1.25, -4) {$\ddots$};
\node[node style, color=orange] (w) at (-1.8, -4.75) {\small$IF^{\pi^1}(\bm{\varphi^{F,\textbf{eff}}_2})$};
\node[node style ] (p) at (-2.25+7, -3.5) {$\varphi^{\pi_2}_{2.1}$};
\node[node style ] (q) at (-1.25+7, -3.5) {$\varphi^{\pi^2}_{2.2}$};
\node[node style ] (r) at (-2.25+7, -4) {$\bm{\varphi^{\pi_2}_{2.3}}$};
\node[node style ] (r) at (-1.25+7, -4) {$\ddots$};
\node[node style, color=red] (z) at (-1.8+7, -4.75) {\small$IF^{\pi^2}(\bm{\varphi^{F,\textbf{eff}}_2})$};
\draw (-5.6,-3) -- (9.5,-3);

\node[node style, color=orange] (v) at (-4, -1.25) {\small$IF^{\pi^1}(\varphi^F_1)$};
\node[node style, color=orange] (w) at (-1.8, -1.25) {\small$IF^{\pi^1}(\bm{\varphi^{F,\textbf{eff}}_2})$};
\node[node style, color=orange] (x) at (0, -1.25) {\small $IF^{\pi^1}(\varphi^F_3)$};
\node[node style, color=black] (ab) at (-1.5, -2) {{\small Influence Functions in Model} $\mathcal{M}_{\pi^1= (\pi^1_1, \pi_2^1)}$};

\node[node style, color=red] (y) at (-4+7, -1.25) {\small$IF^{\pi^2}(\varphi^F_1)$};
\node[node style, color=red] (z) at (-1.8+7, -1.25) {\small$IF^{\pi^2}(\bm{\varphi^{F,\textbf{eff}}_2})$};
\node[node style, color=red] (aa) at (7, -1.25) {\small$IF^{\pi^2}(\varphi^F_3)$};
\node[node style, color=black] (ac) at (-1.5+7, -2) {{\small Influence Functions in Model} $\mathcal{M}_{\pi^2 = (\pi^2_1, \pi_2^2)}$};

\node[node style, color=black] (ad) at (9.5, -0.5) {$\dots$};
\node[node style, color=black] (ad) at (9.5, -4) {$\dots$};

\def\dashlength{3pt}
\draw[->] (-0.1,1.15) -- (-3.75, 0.1);
\draw[->] (1.9,1.15) -- (-1.75, 0.1);
\draw[->] (3.9,1.15) -- (0.25, 0.1);
\draw[->] (0.1,1.15) -- (-3.75+7, 0.1);
\draw[->] (2.1,1.15) -- (-1.75+7, 0.1);
\draw[->] (4.1,1.15) -- (0.25+7, 0.1);

\filldraw[fill=violet, fill opacity=0.1, draw=violet, rounded corners=5pt]
  (-1,1) rectangle (7, 2);
\filldraw[fill=orange, fill opacity=0.1, draw=orange, rounded corners=5pt]
  (-4.75, -2) rectangle (-2.75,0);
\filldraw[fill=orange, fill opacity=0.2, draw=orange, rounded corners=5pt]
  (-2.75,-5.5) rectangle (-0.75, -3.5);
\filldraw[fill=orange, fill opacity=0.2, draw=orange, rounded corners=5pt]
  (-2.75,-2) rectangle (-0.75, 0);
\filldraw[fill=orange, fill opacity=0.3, draw=orange, rounded corners=5pt]
  (-0.75,-2) rectangle (1.25, 0);
\filldraw[fill=orange, fill opacity=0.4, draw=orange, rounded corners=5pt]
  (1.25,-2) rectangle (2, 0);
\filldraw[fill=white, fill opacity=0.1, draw=black, rounded corners=5pt]
  (-4.75,-2) rectangle (2, 0);

\filldraw[fill=red, fill opacity=0.1, draw=red, rounded corners=5pt]
  (-4.75+7, -2) rectangle (-2.75+7,0);
\filldraw[fill=red, fill opacity=0.2, draw=red, rounded corners=5pt]
  (-2.75+7,-2) rectangle (-0.75+7, 0);
\filldraw[fill=red, fill opacity=0.2, draw=red, rounded corners=5pt]
  (-2.75+7,-5.5) rectangle (-0.75+7, -3.5);
\filldraw[fill=red, fill opacity=0.3, draw=red, rounded corners=5pt]
  (-0.75+7,-2) rectangle (1.25+7, 0);
\filldraw[fill=red, fill opacity=0.4, draw=red, rounded corners=5pt]
  (1.25+7,-2) rectangle (2+7, 0);
\filldraw[fill=white, fill opacity=0.1, draw=black, rounded corners=5pt]
  (-4.75+7,-2) rectangle (2+7, 0);

\end{tikzpicture}